\documentclass[12pt]{article}

\usepackage{xr}
\makeatletter
\newcommand*{\addFileDependency}[1]{
  \typeout{(#1)}
  \@addtofilelist{#1}
  \IfFileExists{#1}{}{\typeout{No file #1.}}
}
\makeatother

\usepackage{mathrsfs} %

\newcommand{\blind}{1}

\usepackage{appendix}
\usepackage{stackengine}

\usepackage[utf8]{inputenc}
\usepackage{amsmath}
\usepackage{bm}
\usepackage[figuresright]{rotating}
\usepackage{amsthm}
\usepackage{comment}
\usepackage{amssymb}
\usepackage{wasysym}
\usepackage{enumerate}
\usepackage{epsfig}
\usepackage{graphics}
\usepackage[letterpaper]{geometry}
\usepackage{algpseudocode}
\usepackage{algorithm}
\usepackage{color,xcolor}
\usepackage{authblk}
\usepackage{mathrsfs}
\usepackage{enumitem}
\usepackage[authoryear]{natbib}
\usepackage{multirow}
\usepackage{makecell} 
\usepackage{booktabs} 
\usepackage{setspace}
\newcommand{\RNum}[1]{\uppercase\expandafter{\romannumeral #1\relax}}
\newcommand{\rNum}[1]{\lowercase\expandafter{\romannumeral #1\relax}}

\def\btheta{\boldsymbol{\theta}}
\def\bPhi{\boldsymbol{\Phi}}
\def\bOmega{\boldsymbol{\Omega}}

\def\Kcal{\mathcal{K}}
\def\Mcal{\mathcal{M}}
\def\Pcal{\mathcal{P}}

\def\1f{\boldsymbol{1}}
\def\wh{\widehat}

\def\Df{\mathbf{D}}
\def\Kf{\mathbf{K}}
\def\Wf{\mathbf{W}}
\def\Qf{\mathbf{Q}}
\def\Rf{\mathbf{R}}
\def\Uf{\mathbf{U}}
\def\Vf{\mathbf{V}}
\def\Yf{\mathbf{Y}}
\def\Zf{\mathbf{Z}}

\def\wf{\mathbf{w}}
\def\rf{\mathbf{r}}

\def\zf{\mathbf{z}}

\def\Rb{\mathbb{R}}

\usepackage{subfigure}
\usepackage{caption}

\usepackage{tikz}
\usetikzlibrary{shapes,arrows}
\tikzstyle{startstop} = [rectangle,rounded corners, minimum width=3cm, text width = 3cm,minimum height=1cm,text centered, draw=black]
\tikzstyle{io} = [trapezium, trapezium left angle = 70,trapezium right angle=110,minimum width=4cm,minimum height=1cm, text width = 3.5cm,text centered,draw=black]
\tikzstyle{process} = [rectangle,minimum width=6cm, minimum height=1cm,text centered,text width =5cm,draw=black]
\tikzstyle{decision} = [diamond,aspect = 3,text centered,draw=black]
\tikzstyle{arrow} = [thick,->,>=stealth]
\tikzstyle{arrowdash} = [dashed,->,>=stealth]
\tikzstyle{arrowsum} = [dotted, line width=1.5pt,->,>=stealth]
\tikzstyle{arrowbig} = [thick, line width=2pt,->,>=stealth]

\tikzstyle{item} = [rectangle,rounded corners, minimum width=1.25cm, text width = 1.25cm, minimum height=1cm,text centered, draw=black]
\tikzstyle{itemlong} = [rectangle,rounded corners, minimum width=5.75cm, text width = 5.75cm, minimum height=1cm,text centered, draw=black]

\newtheorem{remark}{Remark}
\newtheorem{prop}{Proposition}
\newtheorem{thm}{Theorem}
\newtheorem{defn}{Definition}

\def\ord{\mathsf{ord}}
\def\add{\mathsf{add}}
\def\multi{\mathsf{multi}}
\def\BIC{\mathsf{BIC}}
\def\Cov{\mathrm{Cov}}
\def\Exp{\mathsf{Exp}}
\def\EEzGP{\mathsf{EEzGP}}
\def\EzGP{\mathsf{EzGP}}
\def\ctGP{\mathsf{ctGP}}
\def\Gau{\mathsf{Gau}}
\def\Linear{\mathsf{Linear}}
\def\Corr{\mathrm{Corr}}
\def\LOOCV{\mathrm{LOOCV}}
\def\LOOCVL{\mathsf{LOOCV}_{\mathrm{l}_2}}
\def\LOOCVlike{\mathsf{LOOCV}_{\mathrm{log-lik}}}

\def\MSel{\mathsf{MSel}}
\def\MAvr{\mathsf{MAvr}}
\def\RRMSE{\mathrm{RRMSE}}

\newcounter{subassumption}[assumption]

\makeatletter
\renewcommand{\p@subassumption}{\theassumption}
\makeatother

\newcounter{sublemma}[lemma]
\renewcommand{\thesublemma}{(\textit{\alph{sublemma}})}
\makeatletter
\renewcommand{\p@sublemma}{\thelemma}
\makeatother

\newcounter{subsublemma}[sublemma]

\makeatletter
\renewcommand{\p@subsublemma}{\thelemma\thesublemma}
\makeatother

\usepackage[CJKbookmarks=true,
bookmarksnumbered=true,
bookmarksopen=true,
colorlinks=true,
citecolor=blue,
linkcolor=blue,
anchorcolor=blue,
urlcolor=blue]{hyperref}
\usepackage{enumitem}

\begin{document}

\def\spacingset#1{\renewcommand{\baselinestretch}%
{#1}\small\normalsize} \spacingset{1}

\if1\blind
{
  \title{\bf A General Framework For Modeling Gaussian Process with Qualitative and Quantitative Factors\thanks{This version is accepted for publication in \emph{Technometrics}.}}
  \author{Linsui Deng and C. F. Jeff Wu$^{ \#}$\\
    School of Data Science, The Chinese University of Hong Kong, Shenzhen, China
}
\def\thefootnote{\#}\footnotetext{Corresponding author. Email: jeffwu@cuhk.edu.cn}
\def\thefootnote{\arabic{footnote}}
  \maketitle
} \fi

\if0\blind
{
  \bigskip
  \bigskip
  \bigskip
  \begin{center}
    {\LARGE\bf A General Framework For Modeling Gaussian Process with Qualitative and Quantitative Factors}
\end{center}
  \medskip
} 
\fi

\small 
\begin{abstract}
\end{abstract}
\noindent%
Computer experiments involving both qualitative and quantitative (QQ) factors have attracted increasing attention. Gaussian process (GP) models have proven effective in this context by choosing specialized covariance functions for QQ factors. In this work, we extend the latent variable-based GP approach, which maps qualitative factors into a continuous latent space, by establishing a general framework to apply standard kernel functions to continuous latent variables. This approach provides a novel perspective for interpreting some existing GP models for QQ factors and introduces new covariance structures in some situations. The ordinal structure can be incorporated naturally and seamlessly in this framework. Furthermore, the Bayesian information criterion and leave-one-out cross-validation are employed for model selection and model averaging. The performance of the proposed method is comprehensively studied on several examples.\\

\vfill
\newpage

\spacingset{1.8} 
\section{Introduction}\label{sec:intro}

Computer experiments have attracted increasing attention in science, engineering, and business due to their ability to model complex systems. However, the high computational cost of running these simulations often necessitates the use of surrogate models or emulators. Among these, Gaussian process (GP) modeling has emerged as a powerful approach because it can approximate the behavior of simulations accurately and efficiently \citep{santner2003design}. Recent developments have extended GP modeling to support a variety of input types, such as probability distributions \citep{bachoc2017gaussian} and functions \citep{li2022gaussian}. 

In many applications, the inputs of a computer experiment involve both quantitative and qualitative factors, commonly referred to as QQ inputs. For instance, in embankment system design, the inputs include one quantitative variable (shoulder distance from the centerline) and three qualitative variables (construction rate, Young’s modulus of columns, and reinforcement stiffness) \citep{LIU2015567,Deng:2017wf}. Similarly, modeling the thermal dynamics of a data center requires consideration of qualitative factors such as diffuser location, return air vent location, and rack heat load nonuniformity, along with quantitative factors like rack temperature rise, rack heat load, and total diffuser flow rate \citep{schmidt2005challenges,Qian:2008ts}. These examples show the importance of developing GP models that work well with QQ inputs. 

An essential step of GP modeling is the construction of the covariance function. In recent years, a number of covariance structures have been proposed and investigated to improve prediction accuracy for handling computer experiments with QQ inputs \citep{Qian:2008ts,Zhou:2011ty,Deng:2017wf,Zhang:2020gp,Roustant2020,Garrido2020,Tao2021vg,xiao2021ezgp,lin2024category}, as reviewed in Section~\ref{sec:existing}. Nevertheless, there still lacks a general framework that ties these approaches together.

In this paper, we propose a general latent-variable-based framework for GP modeling with QQ inputs, built upon the latent variable approach proposed by \cite{Zhang:2020gp}. We show that this framework can include many existing covariance structures and allows a systematic development of new ones. This is achieved by applying different kernel functions to the latent variables, such as Gaussian, exponential, and linear kernels. We study identifiability conditions for the latent parameterization for these kernel settings. We further demonstrate how ordinal information can be integrated by imposing constraints on the latent variables. Given a variety of available kernel choices, we employ both leave-one-out cross-validation and the Bayesian information criterion (BIC) for model selection. Finally, we introduce a BIC-based model averaging strategy to robustly combine predictions from models employing different kernels.

The remainder of this paper is structured as follows. Section~\ref{sec:method}  introduces the general framework, establishes its connection to existing approaches, and provides the identifiability condition for the latent parameterization. Section~\ref{sec:ordinal_var} describes how to incorporate ordinal information within the framework. Estimation and prediction procedures are presented in Section~\ref{sec:est_pre}. In Section~\ref{sec:model_s_a}, we describe the model selection and model averaging strategies. Section~\ref{sec:num_comp} provides comprehensive numerical comparisons. Section~\ref{sec:summary}~ summarizes the findings and makes further discussions.

\section{Methodology}\label{sec:method}
\subsection{Framework}
Consider a problem in which the response $Y(\Uf,\Vf)$ has two types of inputs: $\Uf=(u_1,\cdots,u_I)^\top$ are quantitative factors and $\Vf=(v_1,\cdots,v_J)^\top$ are qualitative factors, with each $v_j$ possessing $a_j$ levels, i.e., $v_j\in\{1,2,\cdots, a_j\}$. We model the response using a GP, which is expressed as
$$
Y(\Uf,\Vf) = \mu+G(\Uf,\Vf),
$$
where $\mu$ represents the constant mean term and $G(\Uf,\Vf)$ is a zero-mean GP. The primary goal is to model the covariance between responses $Y(\Uf,\Vf)$ and $Y(\Uf^\prime,\Vf^\prime)$ corresponding to two distinct inputs $(\Uf,\Vf)$ and $(\Uf^\prime,\Vf^\prime)$. By utilizing the covariance kernel, we can make predictions for new inputs \citep{santner2003design}.

Existing approaches focus on choosing or proposing covariance structures that have some of the attributes: intuitive, interpretable, or computationally efficient \citep{mcmillan1999analysis, Qian:2008ts, Zhou:2011ty, Deng:2017wf}. Motivated by the idea that qualitative variables can be represented by some underlying numerical values, \cite{Zhang:2020gp} proposed that the $j$-th qualitative factor $v_j$ corresponds to a latent vector $\bm{z}^{(j)}_{v_j} \in \mathbb{R}^{l_j}$, where $1 \leq l_j \leq a_j$. Following this formulation, we define the concatenated latent vector $\Zf_\Vf\in\Rb^{\sum_{j=1}^Jl_j}$ as $\Zf_\Vf^\top = \big((\zf^{(1)}_{v_1})^\top,(\zf^{(2)}_{v_2})^\top,\cdots,(\zf^{(J)}_{v_J})^\top\big)$. 
Using this framework, we can state that the response $Y(\Uf, \Vf)$ for input $(\Uf, \Vf)$ follows the same distribution as the response $Y$ for input $(\Uf, \Zf_\Vf)$, i.e., 
$$
Y(\Uf, \Vf) \stackrel{d}{=} Y(\Uf, \Zf_\Vf).
$$
Given the GP assumption, this distributional equivalence holds if and only if their covariance functions are identical. Therefore, we propose to model the covariance function of $Y(\Uf,\Vf)$, the original process of interest, using that of the continuous input $(\Uf,\Zf_\Vf)$ as follows
\begin{equation}\label{equ:covW}
\begin{aligned}
\Cov\left\{Y(\Uf,\Vf),Y(\Uf^\prime,\Vf^\prime)\right\}
=&\Cov\left\{Y(\Uf,\Zf_\Vf),Y(\Uf^\prime,\Zf_\Vf^\prime)\right\}\\
=&\sigma^2\Corr\left\{Y(\Uf,\Zf_\Vf),Y(\Uf^\prime,\Zf_\Vf^\prime)\right\}\\
=&\sigma^2 K_\Uf(\Uf,\Uf^\prime)K_\Zf(\Zf_\Vf,\Zf_\Vf^\prime),
\end{aligned}
\end{equation}
where $\sigma^2$ denotes the variance, and $K_\Uf(\cdot, \cdot)$ and $K_\Zf(\cdot, \cdot)$ are respectively kernel functions for the quantitative factors and latent vectors associated with the qualitative factors. The last identity in \eqref{equ:covW} is based on the assumption that the effects of $\Uf$ and $\Vf$ on $Y$ can be factorized. Both kernels satisfy the normalization condition $K_\Uf(\Uf,\Uf)=1$ and $K_\Zf(\Zf,\Zf)=1$ for all $\Uf\in \Rb^I$ and $\Zf\in\Rb^{\sum_{j=1}^Jl_j}$. This framework is flexible because it can accommodate various kernel functions to capture diverse patterns.

While the assumption of latent vectors may initially appear restrictive, we show that this modeling framework integrates numerous established approaches \citep{Qian:2008ts, Deng:2017wf, Zhang:2020gp, Tao2021vg} as special cases. Furthermore, its inherent generality offers potential for further methodological advancements and applications.

\subsection{Connection with existing approaches}\label{sec:existing}
In the following, we provide a detailed discussion to establish connections between the framework and some existing methods in the literature.

\textit{Multiplicative Linear Kernel.} 
By imposing a multiplicative structure among qualitative variables and adopting the linear kernel for continuous variables \citep{Rojo2018}, the correlation defined by the latent vectors is 
\begin{equation}\label{equ:linear_multi}
    K_\Zf(\Zf_\Vf,\Zf_{\Vf^\prime})
=\prod_{j=1}^J K_j\left(\zf^{(j)}_{v_j},\zf^{(j)}_{v^\prime_j}\right)
=\prod_{j=1}^J\left(\zf^{(j)}_{v_j}\right)^\top\zf^{(j)}_{v^\prime_j}.
\end{equation}

\begin{itemize}
    \item Case I ($l_j=a_j$). The covariance structure proposed by \cite{Qian:2008ts} is defined as 
\begin{equation}\label{equ:covQian}
\Cov\left\{Y(\Uf,\Vf),Y(\Uf^\prime,\Vf^\prime)\right\}=\sigma^2\left\{\prod_{j=1}^J\tau^{(j)}_{v_j,v_j^\prime}\right\}K_{\Uf}(\Uf,\Uf^\prime),
\end{equation}
where $\left(\tau^{(j)}_{v,v^\prime}\right)_{a_j\times a_j}$ are $J$ semi-positive definite matrices with unit diagonal elements (SPDUDE). Here, $\tau^{(j)}_{v,v^\prime}$ represents the correlation between levels $v$ and $v^\prime$ for the $j$th qualitative factor. By using the Cholesky decomposition \citep{Pinheiro:1996tk}, we can represent $\tau^{(j)}_{v_j,v_j^\prime}$ as the product of two column vectors, i.e., 
\begin{equation}\label{equ:Cholesky}
\tau^{(j)}_{v_j,v_j^\prime}=\left(\zf^{(j)}_{v_j}\right)^\top\zf_{v^\prime_j}^{(j)}.
\end{equation}
In this way, the qualitative part in \eqref{equ:covQian} can be equivalently expressed by \eqref{equ:linear_multi}. In addition, \cite{Zhou:2011ty} suggested using hyperspherical parameterization to simplify the computations. 
\item Case II ($l_j<a_j$). \cite{Roustant2020} and \cite{Tao2021vg} explored the case where the length of the latent vector, $l_j$, is shorter than the number of levels, $a_j$, to impose a low-rank structure. This approach significantly reduces the number of hyperparameters and hence alleviates the estimation burden. Building on this, \cite{Tao2021vg} further introduced hyperspherical expressions to address computational challenges. However, the issue of identifiability was not thoroughly discussed (see Section~\ref{sec:identify}).
\item Case III (restricted correlation matrix). To simplify the complexity of the correlation matrix, one may assume specific structures, such as equal correlation considered in \cite{Qian:2008ts}. In our framework, this assumption can be transformed into restrictions on the latent variables. Specifically, it corresponds to a special case where independent noise is permitted, and one-dimensional latent vectors with equal elements are assigned. 
\end{itemize}

\textit{Additive Linear Kernel.} \cite{Deng:2017wf} proposed to model covariance through imposing an additive structure for qualitative factors, which is then multiplied by the correlation attributed to quantitative factors. Our framework has a direct connection with theirs. Suppose $K_\Zf(\cdot,\cdot)$ is the kernel function modified from a first-order additive GP process for continuous variables \citep{Plate1999, Duvenaud2011}, whose individual components employ a linear kernel. To establish this connection, we construct latent vectors satisfying \eqref{equ:Cholesky}. The covariance in \eqref{equ:covW} becomes
\begin{equation}\label{equ:linear_add}
\begin{aligned}
\sigma^2\sum_{j=1}^J K_j\left(\zf^{(j)}_{v_j},\zf^{(j)}_{v_j^\prime}\right) K_{\Uf}(\Uf,\Uf^\prime) 
= &\sum_{j=1}^J \psi_j\sigma^2\left\{\left(\zf^{(j)}_{v_j}\right)^\top\zf^{(j)}_{v^\prime_j}\right\}K_{\Uf}(\Uf,\Uf^\prime)\\
= &\sum_{j=1}^J\sigma_j^2\tau^{(j)}_{v_j,v_j^\prime}K_{\Uf}(\Uf,\Uf^\prime), 
\end{aligned}
\end{equation} 
where $\psi_j$ denotes the weight satisfying $\sum_{j=1}^J\psi_j=1$ and $\sigma_j^2$ represents the variance associated with the $j$th qualitative factor. 
The covariance in equation~\eqref{equ:linear_add} is a special case of the covariance structure proposed by \citet{Deng:2017wf}, which has the form
\begin{equation}\label{equ:covDeng2017}
\sum_{j=1}^J\sigma_j^2\tau^{(j)}_{v_j,v_j^\prime}K_{\Uf,j}(\Uf,\Uf^\prime).    
\end{equation}
In their formulation, the kernel function $K_{\Uf, j}(\Uf, \Uf^\prime)$ varies across different quantitative factors $j$. In contrast, our approach uses a fixed kernel function for all quantitative factors.

\textit{Multiplicative Gaussian Kernel.} \cite{Zhang:2020gp} proposed and evaluated the kernel functions with a multiplicative structure and Gaussian kernel, which is defined as 
\begin{equation*}
K_\Zf(\Zf_\Vf,\Zf_{\Vf^\prime}) =\prod_{j=1}^J\exp\left\{-\left\|\zf_{v_j}^{(j)}- \zf_{v_j'}^{(j)}\right\|^2\right\}.\end{equation*}
Although recognizing the potential to enhance its generality using other kernels, such as power exponential, Matérn, and lifted Brownian kernels, they did not go further to develop details for various choices of kernels. To apply other kernels, comprehensive investigations and empirical validations are essential to ensure their applicability in real-world applications. For instance, the linear kernel has distinct identifiability conditions and requires additional constraints compared to the Gaussian kernel, as will be presented in Section~\ref{sec:identify}.

\textit{Pre-specified Latent Variable.} A popular approach for modeling GP with QQ factors is to encode qualitative variables into numerical vectors. For instance, \cite{Garrido2020} used one-hot encoding vectors of length $a_j$, where the $v$th element is $\boldsymbol{1}(v_j = v)$.  In another instance, \cite{luo2024hybrid} adopted similarity encoding \citep{cerda2018similarity}, which constructs feature vectors from pairwise similarities between qualitative variable levels. Once encoded, standard kernels for continuous variables can then be applied. Therefore, these two methods can be viewed as special cases of the general framework.

\textit{Non-separable Kernel.} As shown in~\eqref{equ:covDeng2017}, \cite{Deng:2017wf} suggested that the smoothness parameters of quantitative factors may vary across different qualitative factors. In addition, \cite{xiao2021ezgp} assumed that these parameters vary across different levels within the same qualitative factor. 
\cite{lin2024category} partitioned the input space into several non-overlapping regions via a qualitative-factor-based tree and fitted a separate GP for each region. These methods could, in principle, be accommodated in our framework by relaxing the factorization assumption in~\eqref{equ:covW} and capturing interactions between QQ factors using non-separable kernels, such as higher-order additive GP \citep{Duvenaud2011}  or tree GP \citep{gramacy2008bayesian}. While such extensions are feasible, the need for careful kernel specification and the substantially greater modeling complexity would likely limit their practical benefits.

To summarize, the relevant methodologies fall into three categories: (i) data‑driven latent variables, summarized in Table~\ref{tab:comparison}; (ii) prespecified latent variables, such as one‑hot and similarity encoding; and (iii) approaches that model interactions between QQ factors. The first two categories can be readily interpreted within the proposed framework, whereas the third cannot be fully explained due to the presence of interactions.

\begin{table}
    \centering
    \begin{tabular}{cccccc}
    \toprule
          & \multicolumn{5}{c}{Kernel Functions}  \\
         \cmidrule{2-6}
        Type & linear & equal & Gaussian & exponential & linear \\
        & full-dim & correlation & low-dim & low-dim & low-dim  \\
         \midrule
       multiplicative &{\color{red}$\blacktriangle $}{\color{orange}$\Circle$}$\CIRCLE$ & {\color{red}$\blacktriangle $}$\CIRCLE$&  {\color{green}$\blacksquare$}{\color{gray}$\square $}$\CIRCLE$ &  {\color{green}$\blacksquare$}$\CIRCLE$ & {\color{gray} $\square$}$\CIRCLE$ \\
       \midrule
       additive & {\color{blue}$\triangle$}$\CIRCLE$& {\color{blue}$\triangle$}$\CIRCLE$ & $\CIRCLE$ & $\CIRCLE$ & $\CIRCLE$ \\
       \bottomrule
    \end{tabular}
    \caption{Comparison of methods explicitly incorporating specific kernel functions and structures. We consider both multiplicative and additive relationships between qualitative variables. Different colors represent different methods: $\CIRCLE$ indicates our framework; {\color{red}$\blacktriangle $} corresponds to the method proposed by \cite{Qian:2008ts}; {\color{orange}$\Circle$} refers to the method by \cite{Zhou:2011ty}; {\color{blue}$\triangle$} represents the approach by \cite{Deng:2017wf}; {\color{green}$\blacksquare$} and {\color{gray} $\square$} denote latent variable-based approaches studied in \cite{Zhang:2020gp} and \cite{Tao2021vg}, respectively.
    }
    \label{tab:comparison}
\end{table}

Finally, another line of research adopts a matrix‑first perspective by structurally parameterizing the correlation matrix of the qualitative factors $\left(\tau^{(j)}_{v,v^\prime}\right)$, which directly extends \cite{Qian:2008ts}. 
\citet{Roustant2020} employed block‑structured correlation matrices to impose group‑level structure on the correlations among levels of each qualitative input. 
\citet{Saves2023} constructed correlation matrices through generalized continuous exponential kernels, which have a generalized form and include the continuous relaxation and Gower distance approaches as special cases. 
Our approach adopts a different perspective by using latent representations for the qualitative factors and induces correlations implicitly through standard kernels defined on the latent space.

\subsection{Identifiability}\label{sec:identify}
Since the relationship between any inputs can be fully characterized by \eqref{equ:covW}, it is crucial to examine the conditions on latent vectors and kernel functions that guarantee uniqueness. To address this, we formalize the concept of parameterization equivalence.

\begin{defn}[Parameterization Equivalence]\label{def:paraEquiv}   
We say the latent parameterization $\left\{\Zf_\Vf\right\}$ under the kernel $K_\Zf(\cdot,\cdot)$ and the latent parameterization $\left\{\Wf_\Vf\right\}$ under the kernel $K_\Wf(\cdot,\cdot)$ are equivalent if $K_\Zf\left(\Zf_\Vf,\Zf_{\Vf^\prime}\right) = K_{\Wf}\left(\Wf_\Vf,\Wf_{\Vf^\prime}\right)$ for all $\Vf, \Vf^\prime\in\times_{j=1}^J\{1,2,\cdots,a_j\}$. 
\end{defn}

To ensure the covariance in \eqref{equ:covW} is well defined, $K_\Zf(\cdot, \cdot)$ is required to be a Mercer kernel \citep{mercer1909,bach2002learning}. Specifically, $K_\Zf(\cdot, \cdot)$ is a function from $\Rb^{\sum_{j=1}^J l_j} \times \Rb^{\sum_{j=1}^J l_j}$ to $\Rb$, and for any $n$ inputs $\Zf_1, \cdots, \Zf_n$, the matrix $\left(K_\Zf(\Zf_i, \Zf_j)\right)_{n \times n}$ must be positive semidefinite. Common examples of Mercer kernels include the Gaussian, exponential, and Matérn kernels. Furthermore, when the kernel exhibits certain separability properties, the linear kernel is very flexible and capable of representing a wide range of kernels.

\begin{thm} \label{thm:general_linear}
Suppose the structure between different qualitative variables is either multiplicative 
\begin{equation}\label{equ:multi_equ}
    K_\Zf(\Zf_\Vf,\Zf_{\Vf^\prime}) = \prod_{j=1}^J K^\Zf_j\left(\zf^{(j)}_{v_j},\zf^{(j)}_{v^\prime_j}\right)
\end{equation}
or additive 
\begin{equation}\label{equ:add_equ}
    K_\Zf(\Zf_\Vf,\Zf_{\Vf^\prime}) = \sum_{j=1}^J \psi_jK^\Zf_j\left(\zf^{(j)}_{v_j},\zf^{(j)}_{v^\prime_j}\right)
\end{equation}
with $\sum_{j=1}^J\psi_j=1$. There always exists a latent parameterization $\left\{\Wf_\Vf\right\}$ with $\wf_{v_j}^{(j)}\in\Rb^{a_j}$ under the kernel $K_\Wf(\cdot,\cdot)$ with $K_j^\Wf(\cdot,\cdot)$ being the linear kernel, that is equivalent to the latent parameterization $\left\{\Zf_\Vf\right\}$ with $\zf_{v_j}^{(j)}\in\Rb^{l_j}$ under the kernel $K_\Zf(\cdot,\cdot)$ with $K_j^\Zf(\cdot,\cdot)$ being any Mercer kernel. 
\end{thm}

\begin{remark}\label{rmk:full_express}
When the kernel is non-separable, all qualitative variables are integrated into a single variable with $\prod_{j=1}^J a_j$ levels, where each level corresponds to a unique combination of the original qualitative variables. In this way, non-separable kernels can be represented by some latent parameterization under linear kernels by applying Theorem~\ref{thm:general_linear}. 
\end{remark}

Theorem~\ref{thm:general_linear} and Remark~\ref{rmk:full_express} hold because we use a linear kernel for the latent representation $\{\Wf_{\Vf}\}$. The proof of Theorem~\ref{thm:general_linear} is provided in the Appendix. The reverse of Theorem~\ref{thm:general_linear} does not hold. For instance, some latent parameterizations under the linear kernel cannot be represented by the Gaussian kernel because the linear kernel can accommodate negative correlations, whereas the Gaussian kernel is restricted to modeling positive correlations only. In other words, when $l_j=a_j$, the linear kernel exhibits greater flexibility compared to other kernels. 
Remark~\ref{rmk:full_express} addresses the case where kernels among qualitative variables are non-separable. A similar approach was proposed in \cite{Oune:2021wp} using the Gaussian kernel. However, the increased number of levels results in more parameters, which makes parameter estimation more difficult.

Although such a reparameterization always exists, if the true underlying kernel structure follows or approximately follows a Gaussian kernel with $l_j<a_j$, it may be possible to recover the structure using fewer parameters, which can reduce redundancy and variability while improving computational efficiency.

With a fixed kernel function $K_\Zf(\cdot,\cdot)$, different latent vectors can produce identical covariance structures. Hence, examining the uniqueness of latent vectors and identifying the essential components that define the covariance is important. This identifiability issue is dependent on the specific kernel function employed.

To begin with, understanding the parameterization problem from a geometric perspective provides valuable insights. For the linear kernel, equivalence under orthogonal transformations corresponds to isometries in the inner product space, where angles and distances are preserved. Rotations about the origin and reflections across any plane passing through the origin maintain the relationships between vectors. For Gaussian kernels (later extended to isotropic kernels), equivalence corresponds to isometries in the distance space, where distances remain invariant. In addition to rotation and reflection, translations also preserve the relative distances between vectors.

Below, we provide an identifiability condition for the linear and isotropic kernels when $I=0$ and $J=1$, where $I$ is the total number of quantitative factors and $J$ is the total number of qualitative factors. An isotropic kernel is a kernel function $K(\mathbf{x}, \mathbf{x}')$ that depends only on the relative Euclidean distance between two inputs, i.e., $K(\mathbf{x}, \mathbf{x}') = \Kcal(d)$, where $d = \|\mathbf{x} - \mathbf{x}'\|_2$ and kernel generating function $\Kcal(\cdot)$ is a nonnegative and monotonically decreasing. Here, we take $I=0$ (i.e., there is no quantitative factor) because the identifiability issue arises only from the latent representations associated with qualitative factors. For the case $J>1$, under the multiplicative structure in \eqref{equ:multi_equ} or the additive structure in \eqref{equ:add_equ} considered in Theorem~\ref{thm:general_linear}, identifiability can be established by examining each qualitative factor individually.

\begin{prop}[Identifiability for linear kernel]\label{prop:unique_linear}
    Consider the latent parameterization $\{\Zf_\Vf\}$ such that $\big(\zf_{1}^{(1)},\cdots,\zf_{l_1}^{(1)}\big)$ is full rank and $\|\zf_{v}\|_2=1$ for all $1\leq v\leq a_1$.\\ (a) There always exists a unique parameterization $\{\Wf_\Vf\}$ satisfying $w_{v,l}^{(j)}=0$ for all $l> v$ and $w_{v,v}^{(j)}> 0$ for $1\leq v\leq l_1$ that is equivalent to $\{\Zf_\Vf\}$ under the linear kernel.\\ (b) $\{\Wf_\Vf\}$ can be deployed to hyperspherical coordinates through 
    $$w_{v,l}^{(1)} =  \cos\left(\theta^{(1)}_{v,l}\right) \prod_{\iota=1}^{l-1} \sin\left(\theta^{(1)}_{v,\iota}\right) \ \text{for}\ 1\leq l<l_1 -1 \ \text{and}\  
w_{v,l}^{(1)} =  \prod_{\iota=1}^{l} \sin\left(\theta^{(1)}_{v,\iota}\right)  \ \text{for}\ l=l_1 -1,  $$
with the constraint $0\leq \theta^{(1)}_{v,l}\leq \pi$ for $1\leq l< \min(v,l_1)$, $0\leq \theta^{(1)}_{v,l_1-1}\leq 2\pi$ for $v\geq l_1$ and $\theta^{(1)}_{v,l}=0$ for $ v\leq l$, where $v\in\{1,\cdots, a_1\}$. 
\end{prop}

\begin{prop}[Identifiability for isotropic kernel]\label{prop:unique_isotropic}
    Consider the latent parameterization $\{\Zf_\Vf\}$ such that $\big(\zf_{2}^{(1)}-\zf_{1}^{(1)},\cdots,\zf_{l_1+1}^{(1)}-\zf_{1}^{(1)}\big)$ is full rank. There always exists a unique parameterization $\{\Wf_\Vf\}$ satisfying $w_{v,l}^{(1)}=0$ for all $l\geq v$ and $w_{v,v-1}^{(1)}> 0$ for $2\leq v\leq l_1+1$ that is equivalent to $\{\Zf_\Vf\}$ under the isotropic kernel.  
\end{prop}

The proofs of Propositions~\ref{prop:unique_linear} and \ref{prop:unique_isotropic} are provided in the Appendix. 
For the linear kernel, the hyperspherical coordinate transformation simplifies the optimization process subject to the norm constraint for the latent vectors (see Lemma~S1). Compared to \cite{Zhou:2011ty}, we allow $l_j\leq a_j$ whereas they restrict $l_j=a_j$. Consequently, additional attention is paid to determine the range of the angles. Specifically, their range of angles is $[0,\pi]$, whereas we require the additional condition $0\leq \theta^{(1)}_{v,l_1-1}\leq 2\pi$ for $v\geq l_1$.  
Note that \cite{Zhang:2020gp} and \cite{Yerramilli:2023vq} claimed that the Gaussian kernel is invariant under translation and rotation. However, as stated in Proposition~\ref{prop:unique_isotropic}, to ensure uniqueness also requires translation invariance. For instance, when $l_j=2$ and $a_j>3$, we require $w_{3,2}^{(1)}>0$. Finally, our approach can extend the scope to encompass any isotropic kernel.

\section{Ordinal variable}\label{sec:ordinal_var}
For ordinal variables, adjacent levels are generally expected to exhibit closer relationships. \cite{luo2024hybrid} treats ordinal variables as nominal when the number of levels is small, and as continuous variables constrained to integer values when large. This empirical rule is straightforward to apply but may either overlook the underlying ordinal structure or be applicable only when the ordinal variable takes integer values. \cite{Qian:2008ts} proposed two approaches for modeling ordinal correlations: one applies constraints to the correlation matrix, while the other transforms the ordinal scale into a continuous variable and defines correlations based on the transformed values. Although the two approaches provide valuable conceptual schemes, they do not offer specific algorithms for practical implementation. \cite{Roustant2020} developed upon the second approach of \cite{Qian:2008ts} by applying a cosine kernel to the distances between transformed values.  However, their method handles nominal and ordinal variables separately and regards them as two distinct types.

Our framework is directly applicable to ordinal variables by imposing order constraints on latent variables. Specifically, for the isotropic kernel, we consider a one-dimensional latent vector $ \mathbf{z}^{(j)} $ (or scalar latent variable; we use the same notation for simplicity) and require that $0 = z_{1,1}^{(j)} \leq z_{2,1}^{(j)} \leq \cdots \leq z_{a_j,1}^{(j)}$. 
This constraint ensures that the relative distances between the latent variables preserve the ordinal information. 
As discussed in Section~\ref{sec:identify}, under the linear kernel, the correlation between two latent vectors is determined by the angle between them. This observation motivates our use of angles in hyperspherical coordinates to encode ordinal information. Specifically, we define $z_{v,1}^{(j)}=\cos\left(\theta_{v,1}^{(j)}\right)$ and $z_{v,2}^{(j)}=\sin\left(\theta_{v,2}^{(j)}\right)$ with the constraint  $0=\theta_{1,1}^{(j)}\leq\theta_{2,1}^{(j)}\leq\cdots\leq\theta_{a_j,1}^{(j)}\leq \pi$, where ${\theta^{(j)}_{v,1}}$ denotes the angles associated with the $v$th levels of the $j$th qualitative variable. This transformation effectively captures one-dimensional ordinal structure within a two-dimensional latent space.

To simplify optimization under ordinal constraints, we reparameterize the ordinal structure using non-negative increments. Specifically, for the isotropic kernel, we define $ z_{v,1}^{(j)}=\sum_{\iota=1}^{v}\Delta^{(\zf,j)}_{\iota} $, where $\Delta^{(\zf,j)}_{1}=0$ and $\Delta^{(\zf,j)}_{\iota}\geq0$, $2\leq\iota\leq a_j$. In this case, the ordinal constraint is transformed into a box-constrained optimization problem \citep{JSSv076i01}, which can be efficiently solved by the L-BFGS-B algorithm \citep{byrd1995limited}. For notation consistency, we represent  $\theta_{v,1}^{(j)}=\sum_{\iota=1}^{v}\Delta^{(\btheta,j)}_{\iota}$ for the linear kernel, where $\Delta^{(\btheta,j)}_{1}=0$, $\Delta^{(\btheta,j)}_{\iota}\geq0$, $2\leq\iota\leq a_j$, and $\theta_{a_j,1}^{(j)}=\sum_{\iota=1}^{a_j}\Delta^{(\btheta, j)}_{\iota}<\pi$. The parameters are optimized using an adaptive barrier algorithm \citep[Chapter~16.3]{lange1999numerical}. A summary of the reparameterizations, along with the identifiability conditions and the number of parameters, is shown in Table~\ref{tab:kernel}. 

As shown in Table~\ref{tab:kernel}, the latent representations $\{\Zf_\Vf\}$ are uniquely determined by these reparameterized parameters, provided under the identifiability conditions. For convenience, we collectively denote these reparameterized parameters as $\{\Omega_\Vf\}$, with specific forms based on the kernel and variable type: for a linear kernel, $\big\{\Omega_\Vf\big\} = \big\{\theta_{\nu,\iota}^{(j)}\big\}$ (nominal) or $\big\{\Omega_\Vf\big\} = \big\{\Delta_{\iota}^{(\btheta,j)}\big\}$ (ordinal); for an isotropic kernel, $\big\{\Omega_\Vf\big\} = \big\{\Zf_{\Vf}\big\}$ (nominal) or $\{\Omega_\Vf\} = \big\{\Delta_{\iota}^{(\zf,j)}\big\}$ (ordinal). Imposing the identifiability conditions is achieved by restricting $\{\Omega_\Vf\}$ to a specific region $\Mcal_\Omega$.
    
\begin{table}[htbp]
    \centering
    \renewcommand{\arraystretch}{1.4}
    \resizebox{\textwidth}{!}{
    \begin{tabular}{@{}cccccc@{}}
        \toprule
        {Kernel} & {Variable} & 
      \multirow{2}{*}{$K_j\left(\zf_{v}^{(j)},\zf_{v^\prime}^{(j)}\right)$} &  {Reparameterization of} & Identifiability Condition &  {Num. of Para.} \\
        {Type} & {Type} &  &$z_{v,l}^{(j)}$  & $\Mcal_\Omega$ &  $\Pcal_j(l_j;a_j)$\\
        \midrule

        \multirow{5}{*}{Linear} 
        & Nominal & $\left(\zf_{v}^{(j)}\right)^\top \zf_{v^\prime}^{(j)}$   
        & $\begin{aligned}
            \cos\left(\theta^{(j)}_{v,l}\right) \prod_{\iota=1}^{l-1} \sin\left(\theta^{(j)}_{v,\iota}\right) &\text{ for } 1 \leq l < l_j-1\\
            \prod_{\iota=1}^{l} \sin\left(\theta^{(j)}_{v,\iota}\right) &\text{ for }  l = l_j-1
            \end{aligned}$  
        & $\begin{aligned}
             0 \leq \theta^{(j)}_{v,l} \leq \pi &\text{ for } 1 \leq l < \min(v,l_j) \\
             0 \leq \theta^{(j)}_{v,l_j-1} \leq 2\pi &\text{ for } v \geq l_j \\
         \theta^{(j)}_{v,l} = 0 &\text{ for } v \leq l 
        \end{aligned}$& $(a_j - 1)(l_j - 1)$ \\
        \cmidrule(l){2-6}

        & $\begin{aligned}
            \text{Ordinal}\\
            (l_j=2)
        \end{aligned}$ &$\left(\zf_{v}^{(j)}\right)^\top \zf_{v^\prime}^{(j)}$ 
        & $\begin{aligned}
            & \cos\left( \sum_{\iota=1}^{v} \Delta_{\iota}^{(\btheta,j)} \right) \text{ for } l=1 \\
            & \sin\left( \sum_{\iota=1}^{v} \Delta_{\iota}^{(\btheta,j)} \right) \text{ for } l=2
        \end{aligned}$
        & $\begin{aligned}
        &\Delta_{1}^{(\btheta,j)} = 0\\
            & \Delta_{v}^{(\btheta,j)} \geq 0  \\
             0 \leq& \sum_{\iota=1}^{a_j} \Delta_{\iota}^{(\btheta,j)} \leq \pi 
        \end{aligned}$ &$a_j - 1$ \\

        \midrule

        \multirow{3}{*}{Isotropic} 
        & Nominal & $\Kcal_j\left( \left\| \zf_{v}^{(j)} - \zf_{v^\prime}^{(j)} \right\|_2 \right)$ & $z_{v,l}^{(j)}$ 
        & $\begin{aligned}
            & z^{(j)}_{v,v-1} > 0\\ 
            &z^{(j)}_{v,l} = 0 \text{ for } l \geq v
        \end{aligned}$ &$ {(2a_j - l_j - 1) l_j}/{2}$\\
        \cmidrule(l){2-6}

        & $\begin{aligned}
            \text{Ordinal}\\
            (l_j=1)
        \end{aligned}$ & $\Kcal_j\left( \left\| \zf_{v}^{(j)} - \zf_{v^\prime}^{(j)} \right\|_2 \right)$ 
        & $ \sum_{\iota=1}^{v} \Delta_{\iota}^{(\zf,j)}$   for  $l=1$
        & $\begin{aligned}
            &\Delta_{1}^{(\zf,j)} = 0\\ 
            &\Delta_{v}^{(\zf,j)} \geq 0
        \end{aligned}$& $a_j - 1$ \\

        \bottomrule
    \end{tabular}
    }
    \caption{Summary of reparameterizations $\{\Omega_\Vf\}$ and ranges $\Mcal_\Omega$ for different qualitative variable types under different kernel structures. }
    \label{tab:kernel}
\end{table}

It is helpful to illustrate how different kernel choices for qualitative variables impose fundamentally different low-dimensional structures. Consider a simple scenario with a single ordinal qualitative variable having $a_1 = 3$ levels and latent dimension $l_1 = 1$. Let $\tau^{(1)}_{v, v^\prime} = K_\Zf\left(\zf^{(1)}_{v}, \zf^{(1)}_{v^\prime}\right)$ denote the kernel-induced correlation between levels $v$ and $v^\prime$ of the qualitative variable. We reparameterize the latent embeddings via non-negative increments as $\theta^{(1)}_{1,1}=\Delta^{(\boldsymbol{\theta},1)}_{1} = 0$, $\theta^{(1)}_{2,1} = \Delta^{(\boldsymbol{\theta},1)}_{2}$, $\theta^{(1)}_{3,1} = \Delta^{(\boldsymbol{\theta},1)}_{2} + \Delta^{(\boldsymbol{\theta},1)}_{3}$, and similarly $z^{(1)}_{1,1} =  \Delta^{(\zf,1)}_{1}=0$, $z^{(1)}_{2,1} = \Delta^{(\zf,1)}_{2}$, $z^{(1)}_{3,1} = \Delta^{(\zf,1)}_{2} + \Delta^{(\zf,1)}_{3}$. Under the linear kernel, the correlations satisfy $\tau^{(1)}_{1,3} = \tau^{(1)}_{1,2} \tau^{(1)}_{2,3} - \big\{1 - (\tau^{(1)}_{1,2})^2\big\}^{1/2} \big\{1 - (\tau^{(1)}_{2,3})^2\big\}^{1/2}$,
which arises from the cosine law on the unit circle in two-dimensional space. In contrast, for the Gaussian kernel, the correlation must satisfy $\tau^{(1)}_{1,3} = \tau^{(1)}_{1,2} \tau^{(1)}_{2,3} \exp\big\{ -2 (\log \tau^{(1)}_{1,2})^{1/2} (\log \tau^{(1)}_{2,3})^{1/2} \big\}$, which follows from solving $\Delta^{(\zf,1)}_{2}$ and $\Delta^{(\zf,1)}_{3}$ given $\tau^{(1)}_{1,2}$ and $\tau^{(1)}_{2,3}$, and then substituting them into the expression for $\tau^{(1)}_{1,3}$. These expressions reveal that even if the pairwise similarities between adjacent levels are identical, the implied similarity between non-adjacent levels can differ substantially depending on the kernel choice. As a result, while \cite{Roustant2020} recommended the linear kernel to capture potential negative correlation and ordinal information simultaneously, this choice may be suboptimal if Gaussian or other kernels are more suitable for capturing the underlying structure. 
In Section~\ref{sec:model_s_a}, we will explore model selection and model averaging techniques to determine the appropriate kernel for analysis and to combine predictions for enhanced performance.

\section{Estimation and prediction}\label{sec:est_pre}

As shown in \eqref{equ:covW}, the kernels for the quantitative and qualitative factors are user-specified. By following standard practice, we employ the Gaussian kernel for the quantitative factors. Denote the unknown parameters within the kernel function as $\bPhi$. We have $$K_\Uf(\Uf,\Uf^\prime)=K_\Uf(\Uf,\Uf^\prime\mid\bPhi)
=\exp\left\{-\sum_{i=1}^I\phi_i(u_i-u_i^\prime)^2\right\}
,$$ 
where $\bPhi=(\phi_1,\phi_2,\cdots,\phi_I)$. Other popular kernels, such as the exponential kernel and the Matérn kernel, can also be used.

In the following, we focus on the category of data‑driven latent variables described in Section~\ref{sec:existing} because the estimated latent variables provide valuable insights into the similarities between different levels. Some of the suggested kernels for qualitative variables are listed in Table~\ref{tab:comparison}. We examine combinations of the multiplicative or additive structure, as defined in \eqref{equ:multi_equ} and \eqref{equ:add_equ}, with different choices of $K_j^\Zf(\cdot,\cdot)$, including Gaussian, exponential, and linear kernels. We consider one and two dimensions for the Gaussian and exponential kernels, and two and three dimensions for the linear kernel. When ordinal information is present, we can also incorporate the corresponding method for comparison.

\textbf{Estimation.} Suppose there are $n$ response values $\Yf=(Y_1,\cdots,Y_n)^\top$ corresponding to input values $\Df = \big\{(\Uf_1,\Vf_1),(\Uf_2,\Vf_2),\cdots,(\Uf_n,\Vf_n)\big\}$. The log-likelihood function up to an additive constant is 
\begin{equation}\label{equ:loglik}
l\left(\mu,\sigma^2,\bPhi,\{\bOmega_\Vf\}\right) 
= -\frac{1}{2}\left\{
n\log(\sigma^2)+\log|\Rf|+\left(\Yf-\mu\1f\right)^\top \Rf^{-1}\left(\Yf-\mu\1f\right)/\sigma^2 
\right\},
\end{equation}
where $\1f$ is a vector of length $n$ with elements one, $|\cdot|$ is the determinant, and $\{\Omega_\Vf\}$ are reparameterized parameters. The latent parameterization $\{\Zf_{\Vf}\}$ is calculated using the information in column 4 of Table~\ref{tab:kernel}. Then,  $\Rf$ is the correlation matrix whose $(i,j)$th element is calculated through $K_\Uf(\Uf_i,\Uf_j\mid\bPhi)K_\Zf(\Zf_{\Vf_i},\Zf_{\Vf_j})$.  In our experiments, the computation of $\Rf^{-1}$ is sometimes numerically unstable due to ill-conditioning. Following \cite{peng2014choice}, we use a nugget term in kriging to improve its conditioning, 
which adds a small positive constant to the diagonal of $\Rf$ and ensure that the smallest eigenvalue is above a prescribed threshold~$\epsilon$. 
Specifically, we consider a sequence of candidate thresholds, e.g., $\epsilon \in \{10^{-1}, 10^{-2}, \cdots, 10^{-8}\}$, 
select the value that minimizes the negative log-likelihood, and then replace $\Rf$ by $\Rf + \delta\,\mathbf{I}_n$, where $\delta = \max\{\,0,\,\epsilon - \lambda_{\min}(\Rf)\,\}$, with $\lambda_{\min}(\Rf)$ denoting the smallest eigenvalue of the original correlation matrix, 
and $\mathbf{I}_n\in\Rb^{n\times n}$ being the identity matrix. Given $\bPhi$ and $\{\Zf_\Vf\}$, $\wh\mu$ and $\wh\sigma^2$ are estimated by 
$$
\wh\mu = \left(\1f^\top\Rf^{-1}\1f\right)^{-1}\1f^\top\Rf^{-1}\Yf
\quad \text{and}\quad
\wh\sigma^2=\left(\Yf-\wh\mu\1f\right)^\top \Rf^{-1}\left(\Yf-\wh\mu\1f\right)/n. 
$$
Plug the estimated mean and variance into the log-likelihood function, and then the remaining parameters are estimated through
$$
\left(
\wh\bPhi,\{\wh\bOmega_\Vf\}
\right)={\arg\min}_{\bPhi,\{\bOmega_\Vf\in\Mcal_{\bOmega}\}} \left\{
n\log(\wh\sigma^2)+\log|\Rf|
\right\},
$$
where $\Mcal_\Omega$ denotes the region describing identifiability conditions, as provided in Table~\ref{tab:kernel}. 

\textbf{Prediction.} We can perform prediction and interpolation in the same manner as in ordinary kriging for quantitative-only variables. Denote the estimated parameters as $\wh\mu$, $\wh\sigma^2$, $\wh\bPhi$, and $\{\wh\bOmega_\Vf\}$. Then, estimated latent representations $\{\wh\Zf_\Vf\}$ are obtained via using the information in column 4 of Table~\ref{tab:kernel}. The estimated correlation matrix, $\wh\Rf$, has its $(i,j)$th entry $K_\Uf(\Uf_i,\Uf_j \mid \wh\bPhi)\, K_\Zf(\wh\Zf_{\Vf_i}, \wh\Zf_{\Vf_j})$. For a new input $(\Uf^\star,\Vf^\star)$, we predict its response and corresponding variance as follows:
$$\wh{Y}(\Uf^\star,\Vf^\star) = \wh{\mu} + \wh\rf^\top\wh\Rf^{-1}(\Yf-\wh\mu \1f)\ \text{and}\ 
s^2(\Uf^\star,\Vf^\star) = \wh\sigma^2\left\{1-\wh\rf^\top \wh\Rf^{-1}\wh\rf +\frac{(\wh\rf^\top\wh\Rf^{-1}\mathbf{1}-1)^2}{\mathbf{1}^\top\wh\Rf^{-1}\mathbf{1}}\right\},$$
where $\wh\rf=(r_1,\cdots,r_n)$ with $r_i=K_\Uf(\Uf_i,\Uf^\star\mid\wh\bPhi)K_\Zf(\wh\Zf_{\Vf_i},\wh\Zf_{\Vf^\star})$.

\section{Model selection and model averaging}\label{sec:model_s_a}
The general framework can induce different models by selecting different kernels and varying the dimension of latent spaces (see Table~\ref{tab:comparison}). Suppose there are $K$ candidate models, denoted by $\{\Mcal_1,\cdots,\Mcal_K\}$. A natural question then arises: How should we determine which kernel to use?

Motivated by the model selection problem in GP with quantitative inputs, we propose two types of criteria to address this question. The first type utilizes leave-one-out cross-validation (LOOCV), which has been commonly used for kernel selection in Gaussian processes with quantitative inputs \citep{Dubrule1983, Rasmussen2006} and for simulator selection \citep{Hung2023}. The second type employs the Bayesian information criterion, BIC \citep{schwarz1978estimating}, which has been shown to provide satisfactory selection performance in a different context \citep{Chen2024}.

\subsection{LOOCV-based model selection}
The LOOCV procedure proceeds as follows. The leave-one-out prediction involves fitting the model while leaving one observation out, then calculating the error based on the model's performance when making predictions using the fitted model. We calculate the LOOCV score for input $i$ of each model $\Mcal_k$ in the prediction step, rather than in the estimation step. Specifically, we calculate the LOOCV score of $\Mcal_k$ as follows: 
$$
\wh{S}_{\Mcal_k} = \frac{1}{n}\sum_{i=1}^n  \wh{S}_{\Mcal_k,(i)},
\quad\text{with}\quad 
\wh{S}_{\Mcal_k,(i)} = L\left(Y_i;\wh\mu_{\Mcal_k,(i)},\wh{\sigma}^2_{\Mcal_k,(i)}\right),
$$
where $\wh\mu_{\Mcal_k,(i)}$ and $\wh{\sigma}^2_{\Mcal_k,(i)}$ are estimated mean and variance using $\Df\backslash \{(\Uf_i,\Vf_i)\}$. The function $L(y;\mu,\sigma^2)$ denotes a user-defined marginal error measure for an observation $y$ under the assumption that $y$ follows a normal distribution with mean $\mu$ and variance $\sigma^2$. According to \cite{Rasmussen2006}, the leave-one-out estimated mean and variance have closed-form solutions, given by
$$\wh\mu_{\Mcal_k,(i)}=Y_i - \left(\wh\Rf^{-1}\Yf\right)_i/\left(\wh\Rf^{-1}\right)_{ii}
\quad\text{and}\quad
\wh{\sigma}^2_{\Mcal_k,(i)} = 1/\left(\wh\Rf^{-1}\right)_{ii},
$$ 
where $(\cdot)_i$ is the $i$th element of vector and $(\cdot)_{ij}$ is the $(i,j)$th element of matrix, respectively. 

We then choose the model minimizing the LOOCV score. In practice, we have two methods using different error measurements $L(y;\mu,\sigma^2)$:
\begin{enumerate}
    \item The $\LOOCVlike$ uses the negative log-likelihood of the normal distribution, defined as 
    $L(y;\mu,\sigma^2) = \{\log(2\pi\sigma^2)\}/2 + {(y - \mu)^2}/{2\sigma^2}.$
    \item The $\LOOCVL$ measures the $l_2$ loss between the predicted value $\wh{\mu}_{\mathcal{M}_k, (i)}$ and the observed value $Y_i$, defined as
    $L(y;\mu,\sigma^2)=(y-\mu)^2.$
\end{enumerate}

\subsection{BIC-based model selection}
The $\BIC_{\MSel}$ method selects the model that minimizes the BIC criterion as the final model. 

When the multiplicative structure is assigned between different qualitative variables, the BIC is defined as
\begin{equation}\label{equ:BIC}
\wh\BIC_{\Mcal_k} =  -2l\left(\wh\mu,\wh\sigma^2,\wh\bPhi,\{\wh\bOmega_\Vf\}\right) +\left(2+I +\sum_{j=1}^J\Pcal_j(l_j;a_j)\right)\ln(n),
\end{equation}
where $\wh\mu,\wh\sigma^2,\wh\bPhi$ and $\{\wh\bOmega_\Vf\}$, $2+I +\sum_{j=1}^J\Pcal_j(l_j;a_j)$ are the number of parameters, including two parameters for mean and variance, $I$ parameters for the scale parameter for the quantitative variables and $\sum_{j=1}^J\Pcal_j(l_j;a_j)$ parameters defined in Table~\ref{tab:kernel}, for the qualitative variables. The first term in \eqref{equ:BIC} evaluates the goodness-of-fit of the model, while the second term imposes a penalty for model complexity.

When an additive structure is employed, the number of parameters increases by $J-1$, as the relative weights $\{\psi_j\}$, shown in \eqref{equ:add_equ}, are also considered unknown parameters and are subject to the constraint $\sum_{j=1}^J \psi_j = 1$.

\subsection{BIC-based model average}\label{sec:modelaverage}
Treating $\{\Mcal_k\}$ as prior models with equal probabilities, we can leverage the results given by different models and hence make predictions more robustly. Following \cite{Claeskens_Hjort_2008}, we can obtain the posterior probability of each model $\Mcal_k$ as follows:
$$
p(\Mcal_k\mid \Df) \propto p(\Df\mid \Mcal_k) p(\Mcal_k) \approx \exp\left(-{\wh{\BIC}_{\Mcal_k}}/{2}\right)p(\Mcal_k).
$$
Therefore, the final model of $\BIC_{\MAvr}$ is a weighted average of different models, where the weights are determined by their BIC values. To derive the prediction and associated uncertainty from the final model,  let 
$\widehat{Y}_{\mathcal{M}_k}$ denote the value predicted by model $\mathcal{M}_k$, 
and let $s^2_{\mathcal{M}_k}$ be the corresponding predictive variance for a new input $(\Uf^\star,\Vf^\star)$. The final predicted value and variance are then given by
$$
\wh{Y}(\Uf^\star,\Vf^\star)=\sum_{k=1}^K w_{\Mcal_k}\wh{Y}_{\Mcal_k}, 
\quad\text{and}\quad s^2(\Uf^\star,\Vf^\star)=\sum_{k=1}^K w_{\Mcal_k}\sqrt{s^2_{\Mcal_k}+\left\{\wh{Y}(\Uf^\star,\Yf^\star)-\wh{Y}_{\Mcal_k}\right\}^2} .
$$
where $w_{\Mcal_k}\propto \exp\left(-{\wh{\BIC}_{\Mcal_k}}/{2}\right)$ and $\sum_{k=1}^K  w_{\Mcal_k} =1$.

\section{Numerical comparisons}\label{sec:num_comp}

In this section, we evaluate the performance of the models in Section~\ref{sec:est_pre} using our framework. 
In the simulation examples, we include three competitive methods for comparison, namely $\EzGP$~\citep{xiao2021ezgp}, $\EEzGP$~\citep{xiao2021ezgp}, and $\ctGP$~\citep{lin2024category}, all of which account for interactions between QQ factors, as discussed in Section~\ref{sec:existing}. In the subsequent computer experiment examples, we focus on further investigating the practical performance and interpretability of our proposed methods. The method names represent Gaussian ($\Gau$), exponential ($\Exp$), and linear ($\Linear$) kernels assigned to the latent variables of qualitative variables. The subscripts denote the dimension of the latent variables or indicate the incorporation of ordinal information ($\ord$) if it exists. We distinguish between the multiplicative ($\multi$) and additive ($\add$) relationships between the qualitative variables using superscripts. For example, $\Gau_{1}^{\multi}$ (or respectively $\Gau_{\ord}^{\multi}$) represents the multiplicative Gaussian kernel with 1-dimensional (or respectively ordinal) latent variables. To mitigate the risk of local optima, the optimization is initialized from 15 random starting points, and the solution yielding the smallest log‑likelihood in \eqref{equ:loglik} is retained.
Building on these base models, our model selection and model averaging strategies, as introduced in Section~\ref{sec:model_s_a}, lead to four methods for evaluation: $\LOOCVlike$, $\LOOCVL$, $\BIC_{\MSel}$, and $\BIC_{\MAvr}$. 

Given the hold-out test points $\Df^{te}=\big\{(\Uf_i^{te},\Vf_i^{te})\big\}_{i=1}^{n_{te}}$, the accuracy is evaluated by the relative root-mean-squared error (RRMSE), which is defined as
$$
\RRMSE = \left\{\frac{\sum_{i=1}^{n_{te}}\left(\wh{Y}(\Uf_i^{te},\Vf^{te}_i)-{Y}(\Uf_i^{te},\Vf^{te}_i)\right)^2}{\sum_{i=1}^{n_{te}}\left({Y}(\Uf_i^{te},\Vf^{te}_i)-\Bar{Y}\right)^2}\right\}^{1/2},
$$
where ${Y}(\Uf_i^{te},\Vf^{te}_i)$ and $\wh{Y}(\Uf_i^{te},\Vf^{te}_i)$ denote the true and predicted values at the input $(\Uf_i^{te},\Vf^{te}_i)$ and $\Bar{Y}=n_{te}^{-1}\sum_{i=1}^{n_{te}}{Y}(\Uf_i^{te},\Vf^{te}_i)$ is the mean of the true responses over the test inputs.

\subsection{Simulation examples}\label{sec:simu}
Following \cite{Zhang:2020gp}, we apply our proposed methods to four real-world engineering models: (i) the beam bending model, (ii) the borehole model, (iii) the output transformerless (OTL) circuit model, and (iv) the piston model. Detailed descriptions of these examples can be found in Section~SI.1. In all examples, qualitative variables are generated from quantitative variables, making them ordinal in nature. For a fair comparison, we employ the same dataset consisting of 30 replicates as used in \cite{Zhang:2020gp}. For each model, the training points were generated using a maximin Latin hypercube design (LHD) \citep{santner2003design}, and 10,000 uniformly distributed test points were used for evaluation. Moreover, the results are insensitive to the choice of experimental design, as illustrated by the additional simulation in Supplementary Section~SII.1, where random sampling and the MaxPro design \citep{joseph2015maximum,joseph2020designing} yield similar performance to the maximin LHD design. The RRMSEs across the 30 replicates are reported in Figure~\ref{fig:4EX}. We summarize the key findings below. 

  \begin{figure}
    \centering
    \includegraphics[width=\linewidth]{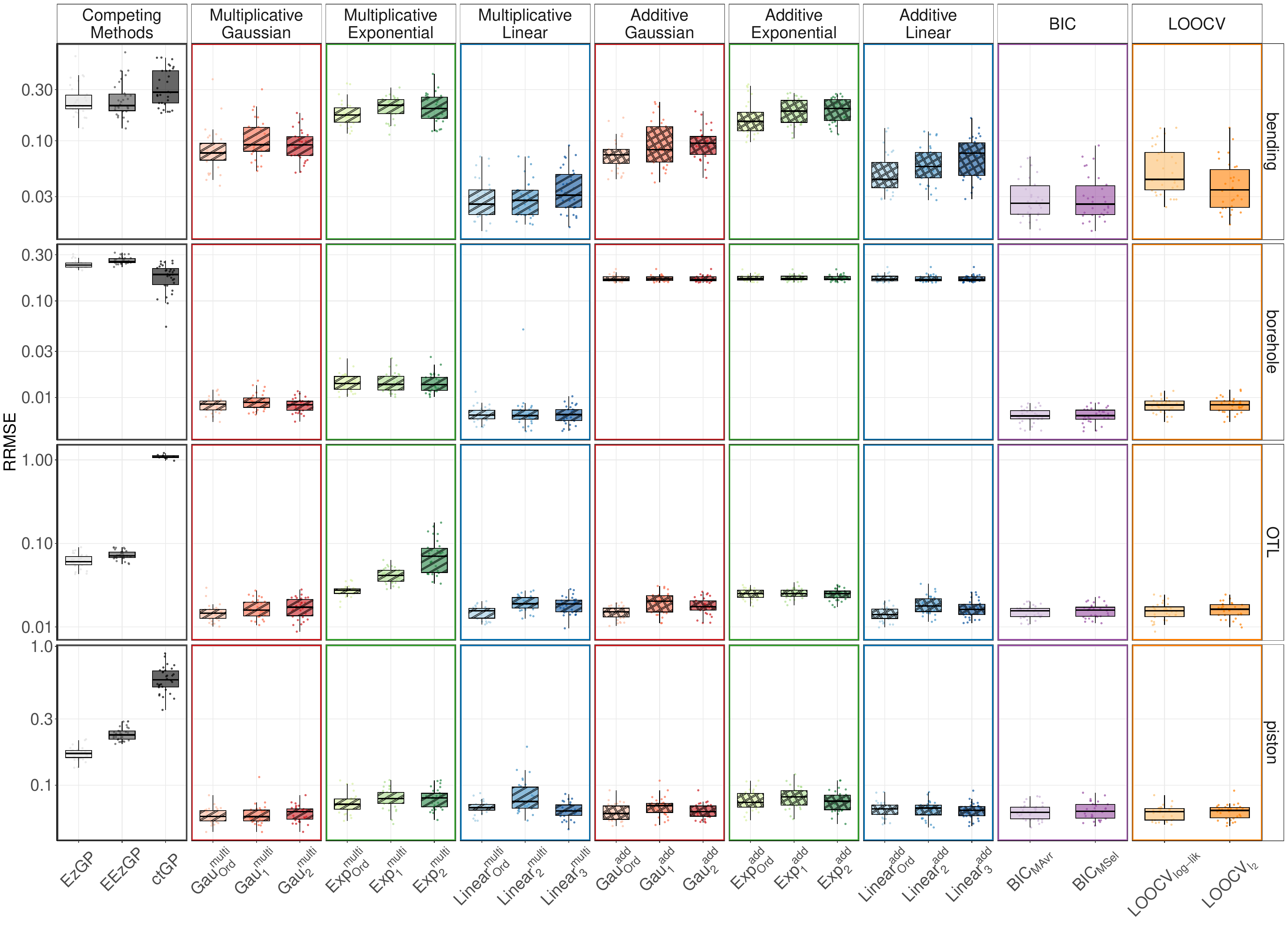}
    \caption{Comparison of RRMSE across different methods and kernel configurations for the beam bending, borehole, OTL circuit, and piston examples. Each boxplot summarizes the results from 30 independent runs with different training sets generated via maximin LHD. All methods were evaluated on the same set of 10,000 uniformly distributed test points.}
    \label{fig:4EX}
\end{figure}

Overall, the three competing methods, especially $\ctGP$, tend to yield higher RRMSE values than the methods under our framework. Our methods equipped with different kernels demonstrate varying strengths across the examples. To better showcase the results, we classify the methods into two categories based on whether they impose additive or multiplicative structures across different qualitative variables. The two classes exhibit distinct performance patterns: (i) the results between the two classes differ significantly in the borehole example; (ii) in the OTL example, the differences between Gaussian, exponential, and linear kernels are less pronounced when an additive structure is imposed, compared to a multiplicative structure. 
Within each class, further comparisons can be made based on the choice of kernel, which determines how the relationships between different levels of each qualitative variable are modeled. 
Methods using the exponential kernel exhibit poor performance, as reflected by the highest RRMSE values. For the other two kernels, methods with the linear kernel perform better than those with the Gaussian kernel in the beam bending, while the former is slightly worse than the latter. For the borehole and OTL examples, they are comparable. These differences can be attributed to the different correlation structures between qualitative variables in the four examples. 

The dimension of the latent vector plays a critical role in determining the performance of methods using the same kernel. For example, in the beam bending example, methods with the linear kernel and a latent vector dimension of $l_j = 2$ ($\Linear^{\multi}_2$ and $\Linear^{\add}_2$) outperform those with $l_j = 3$ ($\Linear^{\multi}_3$ and $\Linear^{\add}_3$). Conversely, in the piston example, the method using the additive Gaussian kernel with $l_j = 2$ ($\Gau_{2}^{\add}$) achieves better performance compared to its counterpart with $l_j = 1$ ($\Gau_{1}^{\add}$). These results demonstrate the importance of selecting an appropriate latent vector dimension. When a low-dimensional latent vector is sufficient to capture the structure of the data, increasing the dimension can introduce additional uncertainty, reduce generalizability, and degrade performance. Conversely, an overly low-dimensional vector may fail to capture the essential relationships and structural complexity. Striking the right balance is essential to ensure the model is both accurate and robust and avoid the pitfalls of over-parameterization or under-specification. 

In this simulation, incorporating the ordinal nature of qualitative variables, when present, generally leads to improved performance. Traditional GP models with QQ inputs often treat ordinal variables as nominal ones. While this approach is convenient, it fails to fully exploit the inherent ordinal structure within the data. In contrast, our methods, which explicitly account for the ordinal structure, consistently have superior performance. Specifically, $\Gau^{\multi}_{\ord}$, $\Gau^{\add}_{\ord}$,  $\Exp^{\multi}_{\ord}$ and $\Exp^{\add}_{\ord}$ in most cases outperform $\Gau^{\multi}_{1}$, $\Gau^{\add}_{1}$, $\Exp^{\multi}_{1}$ and $\Exp^{\add}_{1}$ across all four examples. Moreover, $\Linear^{\multi}_{\ord}/\Linear^{\add}_{\ord}$ outperforms $\Linear_{2}^{\multi}/\Linear_{2}^{\add}$ in the OTL example and performs comparably in the other examples. These findings highlight the limitations of treating ordinal variables as nominal and underscore the importance of leveraging the ordinal structure to enhance modeling accuracy and efficiency. 

Now we investigate the reason behind the improved performance when incorporating ordinal information. Figure~\ref{fig:LV_ORD} visualizes the latent vectors estimated by $\Gau^{\multi}_{1}$, $\Gau^{\multi}_{\ord}$, $\Exp^{\multi}_{1}$, and $\Exp^{\multi}_{\ord}$ in the OTL example where $\Gau^{\multi}_{\ord}$ and $\Exp^{\multi}_{\ord}$ demonstrate significantly better performance than other two.
First, let us focus on the factor $R_f$. The latent vectors estimated by $\Gau_1^{\multi}$ generally follow the order of the levels, with only one exception. This indicates that when using the Gaussian kernel, the ordinal structure is sufficiently strong and can be effectively learned for $R_f$ in this example. This observation reinforces the rationale for using ordinal information. When the exponential kernel is applied, treating the ordinal variable as nominal (i.e., $\Exp_1^{\multi}$) introduces more noise, as the estimated latent vectors exhibit different ordering patterns across replications. 
By incorporating ordinal information, $\Exp_{\ord}^{\multi}$ significantly enhances its predictive power. For the additive kernel, the phenomenon is similar and can be found in Supplementary Figure~S1. Next, let us examine $\beta$, which has six levels and poses a greater challenge. Both $\Gau_1^{\multi}$ and $\Exp_1^{\multi}$ exhibit noisy latent vector estimations. In many replications, most latent vectors corresponding to the six levels estimated by $\Exp_\ord^{\multi}$ are identical, which indicates that the different levels of $\beta$ are difficult to distinguish using this dataset. In such cases, estimating too many parameters for the latent parameterization may lead to overfitting. By comparing the methods that treat the variable as nominal ( $\Exp_1^{\multi}$ and $\Gau_1^{\multi}$) with those that impose an ordinal structure ( $\Exp_{\ord}^{\multi}$ and $\Gau_{\ord}^{\multi}$), it becomes clear that incorporating ordinal constraints has the effect of regularizing parameter estimation and hence enhancing generalizability in this case.

\begin{figure}[htbp]
    \centering
    \subfigure[$R_f$, multiplicative Gaussian kernel.]{
        \includegraphics[width=0.45\linewidth]{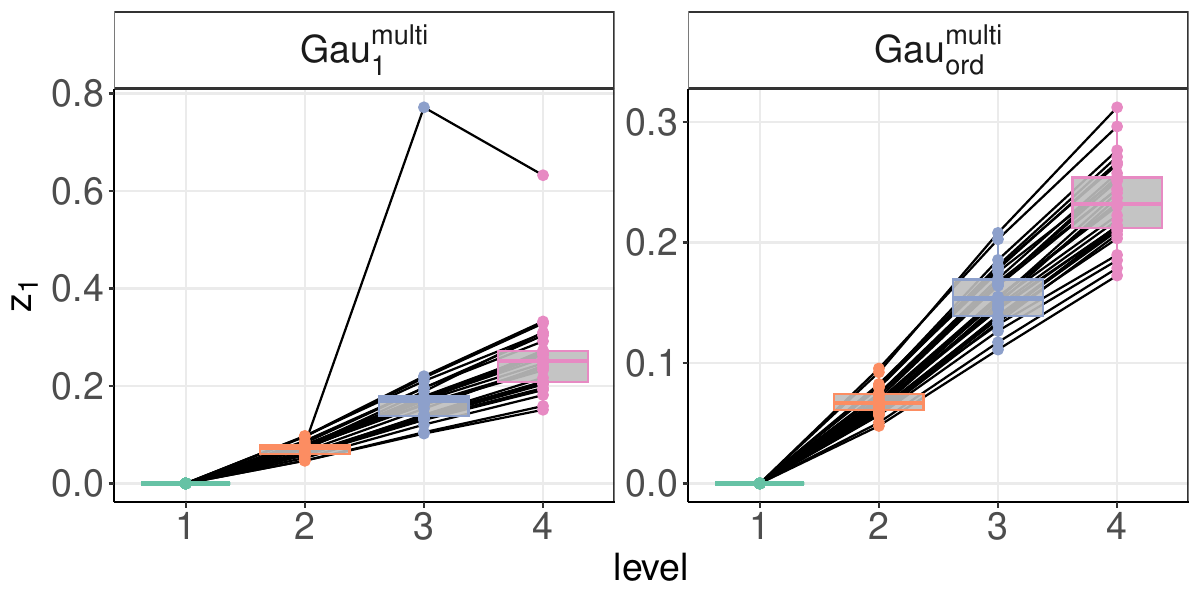}
        \label{fig:lv_ord_rf_gau}
    }
    \subfigure[$R_f$, multiplicative exponential kernel.]{
        \includegraphics[width=0.45\linewidth]{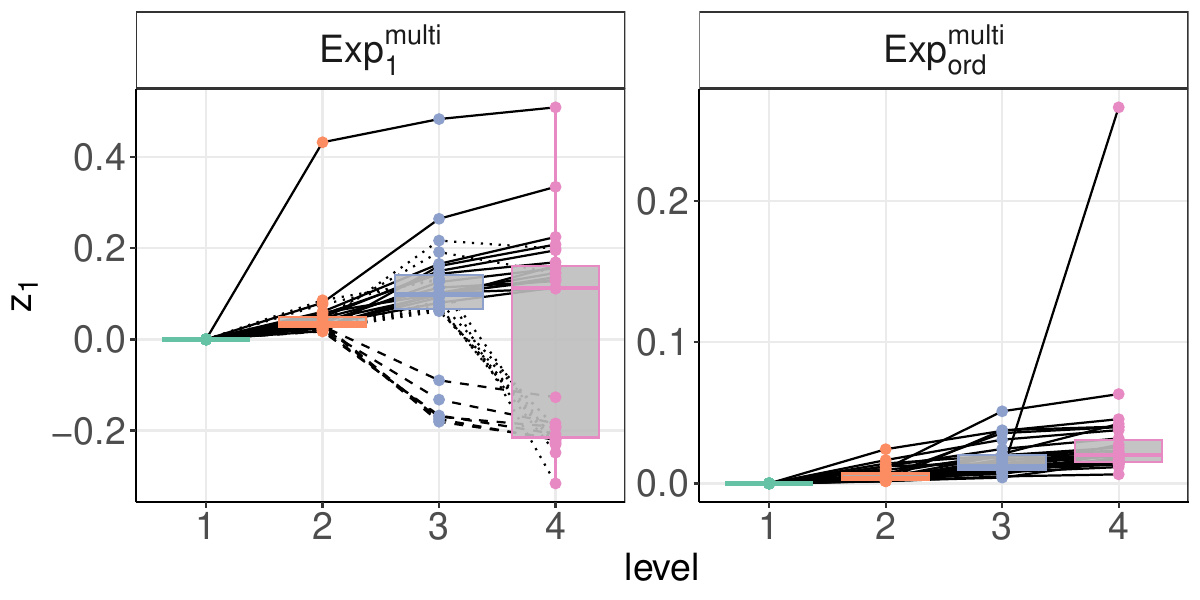}
        \label{fig:lv_ord_rf_exp}
    }
    \vskip\baselineskip 
    \subfigure[$\beta$, multiplicative Gaussian kernel.]{
        \includegraphics[width=0.45\linewidth]{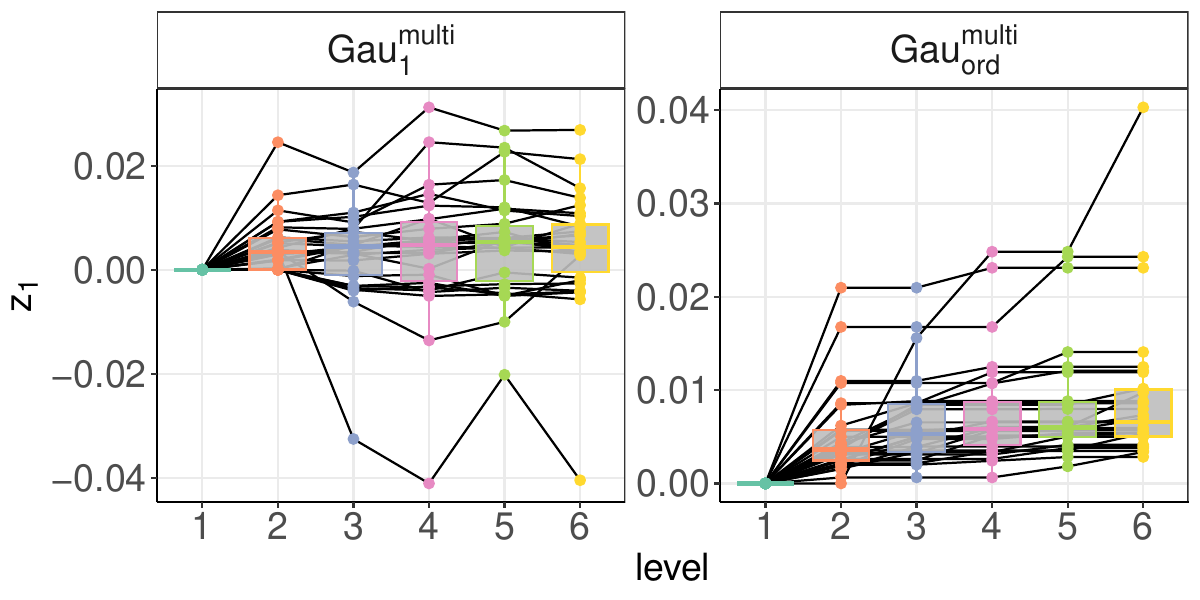}
        \label{fig:lv_ord_beta_gau}
    }
    \subfigure[$\beta$, multiplicative exponential kernel.]{
        \includegraphics[width=0.45\linewidth]{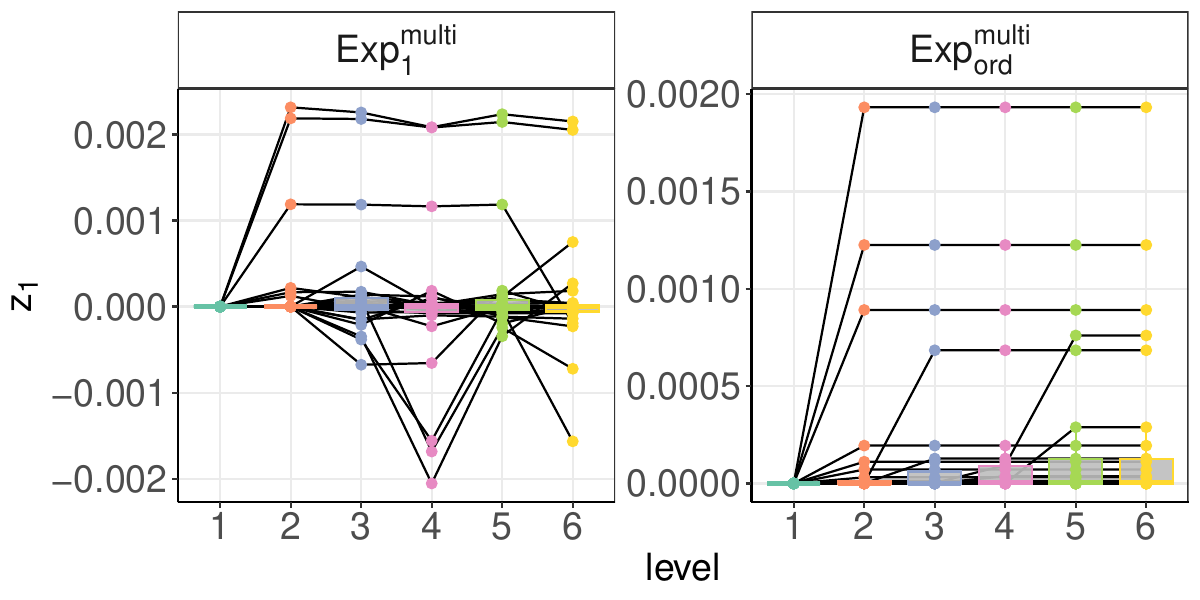}
        \label{fig:lv_ord_beta_exp}
    }
    \caption{The boxplots and scatter points depict the latent vectors $z_1$ for resistance $R_f$ and current gain $\beta$, estimated by $\Gau^{\multi}_{1}$, $\Gau^{\multi}_{\ord}$, $\Exp^{\multi}_{1}$, and $\Exp^{\multi}_{\ord}$, respectively, across 30 replications in the OTL example. Points from the same replication are connected by lines. In the $R_f$ panel for $\Exp^{\multi}_{1}$, different line types indicate different ordering patterns of the estimated latent vectors: solid lines correspond to replications where the original level order is preserved, dashed lines indicate a reversal between the second and third levels, and dotted lines indicate a swap between the third and fourth levels. }
    \label{fig:LV_ORD}
\end{figure}

Since assigning different kernels yields varying performance, it is important to select an appropriate one. We compare three model selection strategies ($\BIC_\MSel$, $\LOOCVL$, and $\LOOCVlike$) and one model averaging strategy ($\BIC_{\MAvr}$) proposed in Section \ref{sec:model_s_a} through their normalized RRMSE ranks, where a lower rank indicates better (i.e., lower) RRMSE performance. As shown in Table~\ref{tab:rank}, these strategies perform satisfactorily in general because they always select methods ranked in the top half. For the model selection strategies, $\LOOCVlike$ achieves relatively stable ranks, which fall within the top 20\% to 30\%, and shows competitive performance in the OTL and piston examples. In contrast, $\BIC_\MSel$ and $\BIC_{\MAvr}$ achieve the best performance in the bending and borehole examples, while their performance is less favorable in the piston example. Here, $\BIC_{\MAvr}$ achieves more stable performance across various scenarios and demonstrates better performance than $\BIC_{\MSel}$ in both the OTL and piston examples.

\begin{table}[h!]
\centering
\caption{Normalized rank of RRMSE for $\BIC_{\MAvr}$, $\BIC_{\MSel}$, $\LOOCVlike$, and $\LOOCVL$ across simulation types. The ranks for $\BIC_{\MSel}$, $\LOOCVlike$, and $\LOOCVL$ are computed among the 18 base models based on ascending RRMSE, whereas the rank for $\BIC_{\MAvr}$ is determined by including its RRMSE alongside the 18 base models. A lower rank implies a lower RRMSE value. The results are summarized as median, mean, and standard deviation (SD) across replications.}
\label{tab:rank}
\resizebox{\textwidth}{!}{%
\begin{tabular}{lcccccccccccc}
\toprule
\multirow{2}{*}{Method} &  \multicolumn{3}{c}{{bending}} &  \multicolumn{3}{c}{{borehole}}&  \multicolumn{3}{c}{{OTL}} &  \multicolumn{3}{c}{{piston}} \\
\cmidrule{2-13}
                & {Median} & {Mean} & {SD} & {Median} & {Mean} & {SD} & {Median} & {Mean} & {SD} & {Median} & {Mean} & {SD} \\
\midrule
$\BIC_{\MAvr}$ & 0.105 & 0.125 & 0.040 & 0.105 & 0.130 & 0.047 & 0.263 & 0.254 & 0.136 & 0.342 & 0.321 & 0.211\\
$\BIC_{\MSel}$ & 0.111 & 0.119 & 0.111 & 0.111 & 0.113 & 0.068 & 0.278 & 0.285 & 0.183 & 0.389 & 0.398 & 0.236\\
$\LOOCVlike$ & 0.222 & 0.248 & 0.106 & 0.250 & 0.237 & 0.076 & 0.278 & 0.280 & 0.178 & 0.222 & 0.302 & 0.211\\
$\LOOCVL$ & 0.111 & 0.157 & 0.117 & 0.278 & 0.246 & 0.071 & 0.278 & 0.296 & 0.179 & 0.333 & 0.370 & 0.209\\
\bottomrule
\end{tabular}}
\end{table}

To sum up, this simulation study highlights the significance of selecting appropriate kernels, optimizing latent vector dimensions, leveraging ordinal structures, and employing effective model selection or averaging strategies. To further examine the prediction accuracy and computational cost for various dimensions, we conduct additional experiments on the borehole example with different discretization degrees. Our methods achieve the lowest RRMSE in moderate dimension. As an example of our approach, $\Gau_{\ord}^{\multi}$ delivers more accurate predictions than $\ctGP$, $\EzGP$, and $\EEzGP$ across all degrees within a reasonable computational time. Interestingly, a trade-off between accuracy and time can be seen. For instance, while $\ctGP$ becomes more accurate with finer discretization, its computational cost increases more rapidly than other methods. Another observation is that $\Linear_{\ord}^{\multi}$ attains higher predictive accuracy than both $\Linear_{2}^{\multi}$ and $\Linear_{1}^{\multi}$ at the expense of a higher computational cost. The detailed setups and results are provided in Supplementary Section~SII.2.

\subsection{A 3D coupled finite element model for embankments}\label{sec:3D}
In this section, we apply various methods to a fully 3D coupled finite element model, which has been rigorously validated for its effectiveness in capturing the deformations and stresses of full-scale embankments \citep{LIU2015567}. 
The corresponding computer experiments involve one quantitative factor and three qualitative factors. The quantitative factor $u_1 \in [0, 14]$ (in $m$) represents the distance from the embankment centerline to the embankment shoulder, taking 29 uniformly spaced values. The three qualitative factors are the embankment construction rate $v_1 \in \{1, 5, 10\}$ (in $m$/month), the Young’s modulus of columns $v_2 \in \{50, 100, 200\}$ (in MPa), and the reinforcement stiffness $v_3 \in \{1578, 4800, 8000\}$ (in kN/$m$). This example is particularly suitable for studying scenarios with a limited number of levels for qualitative variables, as each qualitative factor here has only three levels. 

As described in \cite{Deng:2017wf} and \cite{Kang2020}, for each value of the quantitative factor, a three-level fractional factorial design with nine runs is employed for the qualitative factors, resulting in a total of 261 design points. The test dataset consists of 29 input settings, where $u_1$ takes 29 equally spaced values over interval $[0,14]$, and the qualitative factors are fixed at $(v_1, v_2, v_3) = (5, 100, 4800)$. Among the design points, two points $(u_1, v_1, v_2, v_3) \in \{(14,5,200,8000), (14,10,200,1578)\}$ are identified as having potentially reversed labels and are excluded from the training process (see Supplementary Section~SI.3 for details). Additionally, training on the entire dataset fails. A plausible explanation is that some points that are too close to each other can cause a near singularity when computing $\Rf^{-1}$ in \eqref{equ:loglik}. To address this, subsets of 3, 5, and 7 points are randomly selected at each combination of the qualitative factors, resulting in training sets with 27, 45, and 63 design points, respectively. To evaluate the performance of each method, the RRMSE is calculated using (i) the remaining points in the training dataset and (ii) the independent test dataset.

As shown in Figures~\ref{fig:U3_train_test} and~\ref{fig:U3_test}, increasing the latent dimension generally improves prediction accuracy. Since the number of levels is relatively low (three), increasing the dimension from one to two for Gaussian and exponential kernels or from two to three for Gaussian kernels only requires estimating one additional parameter, which enhances model flexibility with a slight increase in complexity.

\begin{figure}
    \centering
    \includegraphics[width=\linewidth]{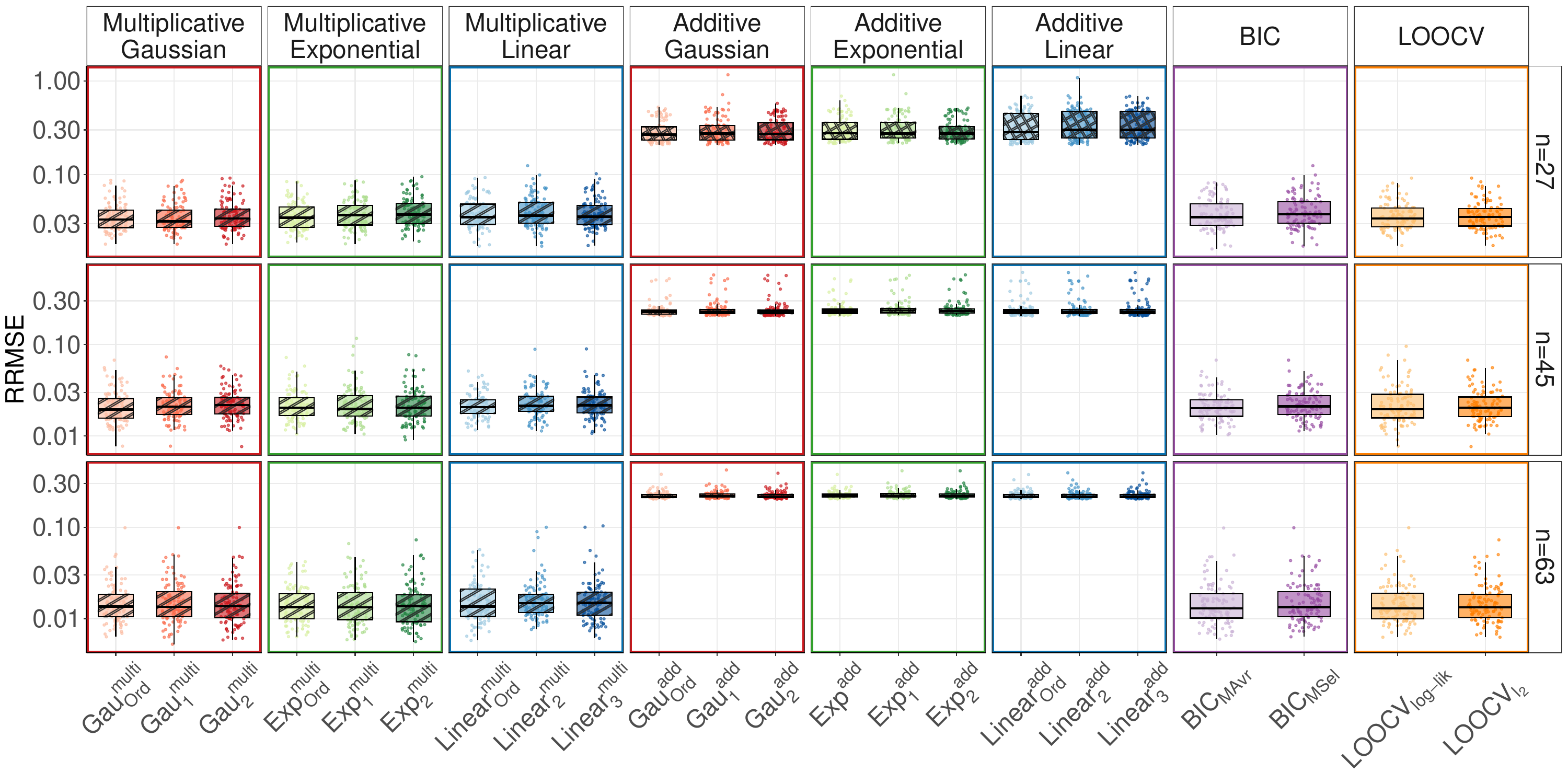}
    \caption{Comparison of RRMSE across different methods and kernel configurations for the 3D coupled finite element model for embankments, evaluated on the remaining points of the training dataset. Each boxplot summarizes results from 100 runs. In each run, the model is trained using a subset of 27, 45, or 63 design points randomly sampled from the training dataset.}
    \label{fig:U3_train_test}
\end{figure}

\begin{figure}
    \centering
    \includegraphics[width=\linewidth]{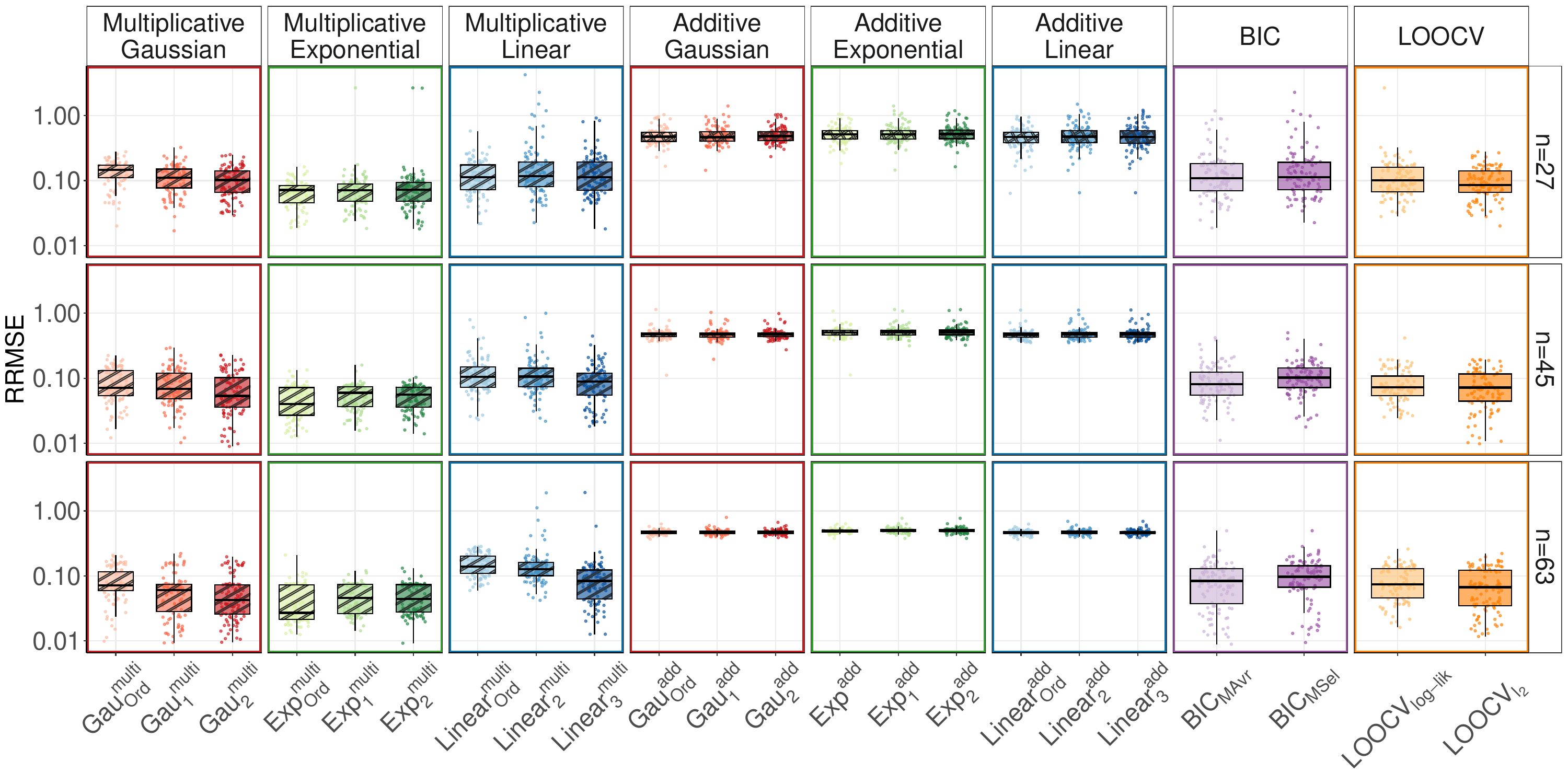}
    \caption{Comparison of RRMSE across different methods and kernel configurations for the 3D coupled finite element model for embankments, evaluated on the independent test dataset. Each boxplot summarizes results from 100 runs. In each run, the model is trained using a subset of 27, 45, or 63 design points randomly sampled from the training dataset.}
    \label{fig:U3_test}
\end{figure}

\subsection{A material design example}\label{sec:material}

We then investigate a material design example focusing on the elastic and mechanical properties of materials \citep{balachandran2016adaptive}. The dataset consists of 223 compounds from the $\mathrm{M}_2\mathrm{AX}$ family, with their elastic properties computed using density functional theory and the planewave/core potential formalism \citep{Cover_2009}. The responses include the bulk modulus, shear modulus, and Young’s modulus. This example involves three nominal variables, each with multiple levels: the M atom has ten levels $\{${Sc}, {Ti}, {V}, {Cr}, {Zr}, {Nb}, {Mo}, {Hf}, {Ta}, {W}$\}$, the A atom has two levels $\{${C}, {N}$\}$, and the X atom has twelve levels $\{${Al}, {Si}, {P}, {S}, {Ga}, {Ge}, {As}, {Cd}, {In}, {Sn}, {Tl}, {Pb}$\}$. In addition, the M, A, and X atoms are associated with three, two, and two quantitative features, respectively, that describe their physical properties. Using the shear modulus as a response, previous studies have demonstrated the importance of incorporating qualitative variables into the prediction \citep{Zhang:2020gp}. Here, we extend this investigation by comparing the performance of various kernel configurations across all three responses. Following the previous approach, we randomly select 200 data points from a total of 223 as the training set, with the remaining 23 data points reserved for testing. The evaluation is repeated 10 times.

\begin{figure}
    \centering
    \includegraphics[width=\linewidth]{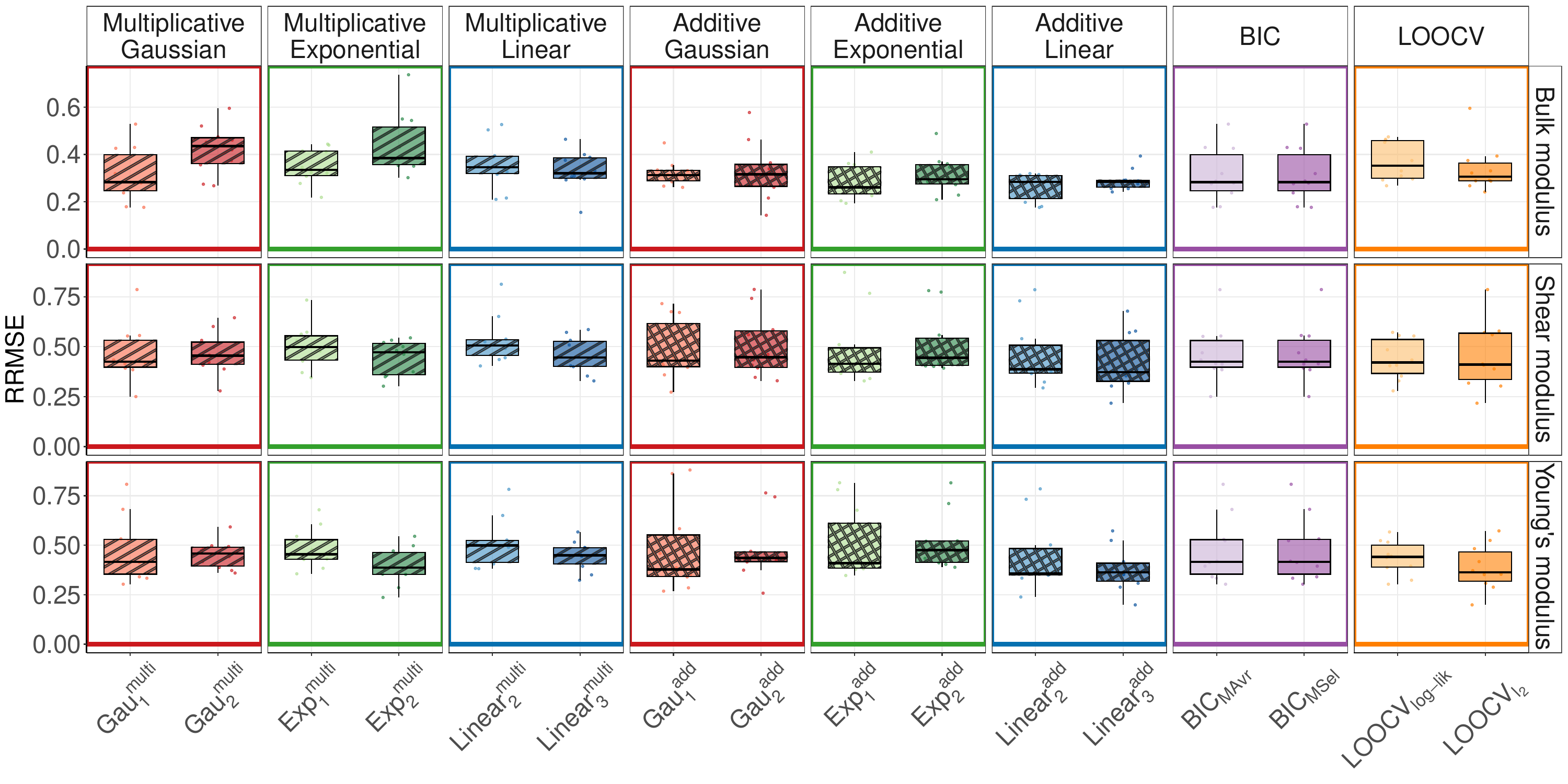}
    \caption{Comparison of RRMSE across different methods and kernel configurations for the material design example. Each boxplot reflects results from 10 resampling runs, where 200 points are used for training and 23 points for testing in each run.}
    \label{fig:ada_design}
\end{figure}

\begin{table}[h!]
\centering
\caption{Summary of RRMSE for 12 base models, three model selection methods, and one model averaging method. Results are presented as median, mean, and standard deviation (SD) across replications. The smallest RRMSE mean and median values are presented in boldface for each response.}
\label{tab:rrmse_3}
\resizebox{\textwidth}{!}{%
\begin{tabular}{ccccccccccccccccccc}
\toprule
&&  \multicolumn{16}{c}{{Method}}  \\
\cmidrule{3-18}
  Modulus&Criterion& \multicolumn{2}{c}{$\Gau^{\multi}$}   &\multicolumn{2}{c}{$\Exp^{\multi}$}   & \multicolumn{2}{c}{$\Linear^{\multi}$}  &  \multicolumn{2}{c}{$\Gau^{\add}$}    & \multicolumn{2}{c}{$\Exp^{\add}$}     &\multicolumn{2}{c}{$\Linear^{\add}$}   & \multicolumn{2}{c}{$\BIC$} & \multicolumn{2}{c}{$\LOOCV$} \\
\cmidrule(r){3-4} \cmidrule(lr){5-6} \cmidrule(lr){7-8} \cmidrule(lr){9-10} \cmidrule(lr){11-12} \cmidrule(lr){13-14} \cmidrule(lr){15-16} \cmidrule(l){17-18}
  && 1-d & 2-d & 1-d & 2-d  & 2-d  &3-d &  1-d & 2-d & 1-d & 2-d  & 2-d  &3-d  & $\MAvr$ & $\MSel$ & $\mathsf{log-lik}$ & $\mathsf{l_2}$\\
\midrule
 &Mean &  0.314 & 0.419 & 0.350 & 0.441 & 0.357 & 0.331 & 0.320 & 0.327 & 0.284 & 0.315 & \textbf{0.263} & 0.291 & 0.314 & 0.314 & 0.372 & 0.339\\
Bulk&Median & 0.283 & 0.435 & 0.335 & 0.384 & 0.345 & 0.320 & 0.312 & 0.316 & \textbf{0.260} & 0.295 & 0.284 & 0.284 & 0.283 & 0.283 & 0.352 & 0.306\\
&SD & 0.115 & 0.104 & 0.075 & 0.132 & 0.104 & 0.083 & 0.054 & 0.123 & 0.073 & 0.081 & 0.058 & 0.045 & 0.115 & 0.115 & 0.083 & 0.101\\
\midrule
&Mean & 0.465 & 0.468 & 0.499 & 0.442 & 0.531 & 0.457 & 0.481 & 0.504 & 0.486 & 0.510 & 0.459 & \textbf{0.418} & 0.465 & 0.465 & 0.437 & 0.454\\
Shear &Median &  0.424 & 0.455 & 0.498 & 0.472 & 0.506 & 0.445 & 0.430 & 0.447 & 0.415 & 0.445 & 0.387 & \textbf{0.373} & 0.424 & 0.424 & 0.421 & 0.410\\
&SD & 0.143 & 0.106 & 0.113 & 0.088 & 0.120 & 0.091 & 0.151 & 0.160 & 0.185 & 0.149 & 0.170 & 0.145 & 0.142 & 0.143 & 0.102 & 0.169\\
\midrule
&Mean &  0.475 & 0.452 & 0.482 & 0.396 & 0.508 & 0.446 & 0.478 & 0.478 & 0.501 & 0.514 & 0.443 & \textbf{0.377} & 0.475 & 0.475 & 0.435 & 0.387\\
Young&Median & 0.416 & 0.458 & 0.454 & 0.387 & 0.499 & 0.449 & 0.377 & 0.436 & 0.409 & 0.475 & \textbf{0.358} & 0.363 & 0.416 & 0.416 & 0.441 & 0.363\\
&SD & 0.163 & 0.071 & 0.101 & 0.095 & 0.126 & 0.075 & 0.225 & 0.157 & 0.182 & 0.143 & 0.179 & 0.109 & 0.163 & 0.163 & 0.086 & 0.114\\
\bottomrule
\end{tabular}
}
\end{table}

The RRMSEs across ten replicates for the three responses are shown in Figure~\ref{fig:ada_design}, with a summary provided in Table~\ref{tab:rrmse_3}. In this example, either $\Linear^{\add}_2$ or $\Linear^{\add}_3$ generally achieves the lowest mean or median RRMSE, with $\Linear^{\add}_3$ exhibiting lower variation. Assigning either an additive or multiplicative kernel results in comparable performance; however, the additive kernel generally performs slightly better.

 To understand this difference, we visualize the relative weights $\{\psi_j\}$ of the additive kernel in Figure~\ref{fig:rel_sig}. Denote the X-atom as the $j_0$th qualitative variable. We observe that its relative weights $\psi_{j_0}$ are close to zero, which indicates that the X-atom contributes minimally to the variation in the response. In other words, there is little distinction between different levels of the X-atom. The same effect can also be achieved by using the multiplicative kernel, provided that $\tau_{v,v^\prime}^{j_0} \approx 1$ in~\eqref{equ:covQian} for all $v, v^\prime \in \{1, \ldots, a_{j_0}\}$, i.e., the correlations between any two levels of the X-atom are approximately one. When the levels are represented via latent variables, this condition says that the latent variables associated with different levels are nearly identical. However, to determine such a structure accurately would require a large amount of data. 
Consequently, when a qualitative variable has little or no effect on the response, the use of the additive kernel may capture this pattern more effectively.

\begin{figure}
    \centering
    \includegraphics[width=\linewidth]{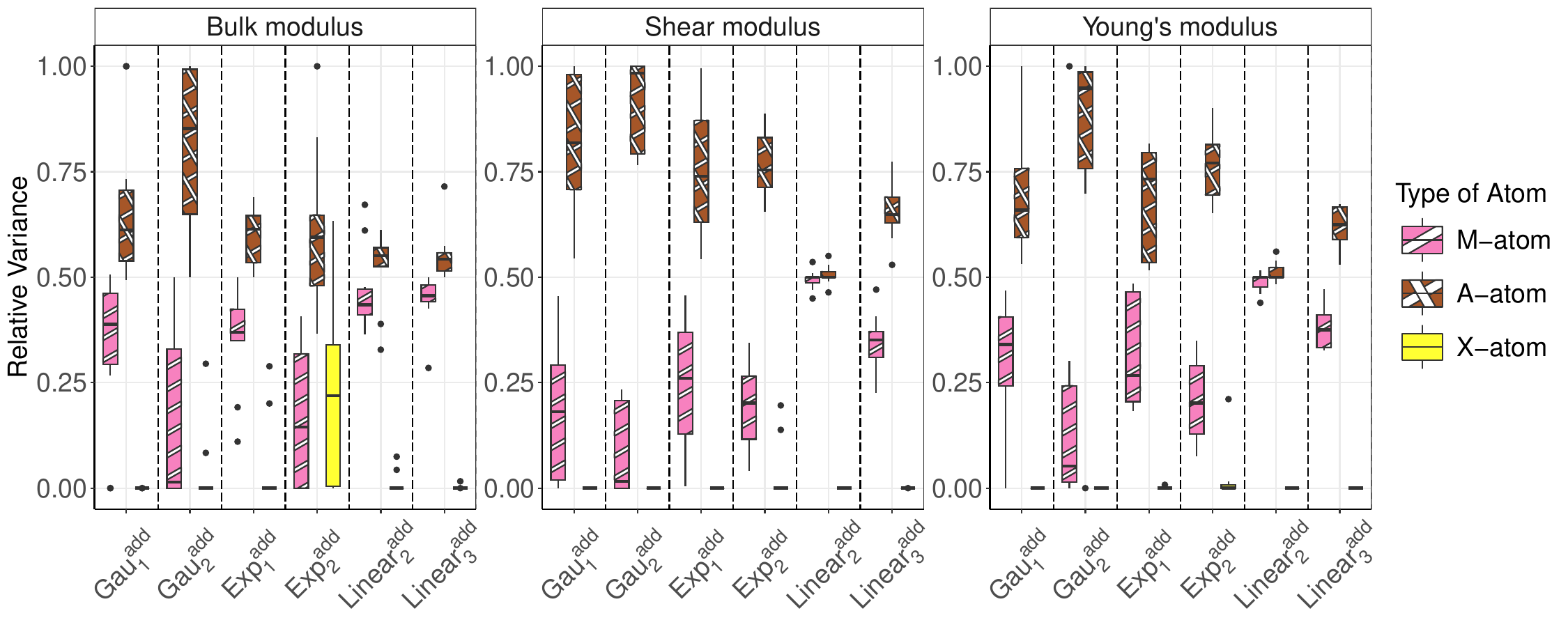}
    \caption{Comparison of relative weights $\psi_j$ defined in \eqref{equ:add_equ} for methods with additive kernels in the material design example. }
    \label{fig:rel_sig}
\end{figure}

\section{Summary remarks and further discussion}\label{sec:summary}
In this paper, we introduce a general framework for modeling QQ factors using GP. The main idea comes from an insightful approach by \cite{Zhang:2020gp}, which maps each qualitative factor into a continuous latent space. We realize that this approach includes many existing models as special cases by employing multiplicative or additive structures and by adopting kernel functions—such as Gaussian, exponential, and linear kernels—for the latent vectors of each qualitative variable. We systematically evaluate the performance of these models. Methods with the linear kernel achieve superior performance in certain cases. Overall, models using multiplicative kernels tend to outperform those using additive kernels, except when some qualitative variables are inactive in predicting the response. Moreover, leveraging ordinal information can improve performance. Finally, both model averaging and model selection strategies effectively identify appropriate models and yield satisfactory predictive accuracy. 

Using this framework, we establish two important connections that enhance both conceptual understanding and practical implementation. 
First, our framework bridges the gap between quantitative input-only modeling and QQ input modeling. This connection enables the extension of techniques from quantitative input-only Gaussian processes to situations involving both quantitative and qualitative inputs. A common approach is to transform qualitative levels into dummy variables and then apply kernels designed for quantitative variables. While this method is straightforward to implement, it often complicates the interpretation of learned scale parameters. Beyond the approach taken by \cite{Zhang:2020gp}, we take a further step by studying its compatibility with other kernels for quantitative inputs. This can facilitate future work to develop more complex kernels for QQ inputs. However, careful attention must be paid to improve interpretability and tackle computational challenges arising from the discrete nature of qualitative inputs.

Second, we demonstrate that this framework unifies many existing approaches by choosing different kernels, as shown in Section~\ref{sec:existing}. For these models, we propose model selection or model averaging procedures, the effectiveness of which is confirmed by the simulation and numerical results in Section~\ref{sec:num_comp}. Besides unifying existing approaches, we find that assigning a low-dimensional linear kernel is equivalent to imposing a low-rank structure on the correlation matrix across different levels. The effectiveness of low-rank structures has been well demonstrated in many fields, such as economics \citep{fan2008high}, epidemiology \citep{zhong2024reduced}, and engineering \citep{chang2021reduced}.

Finally, our framework can easily handle ordinal variables by imposing constraints on the latent vectors to incorporate the ordinal information. To facilitate computation, we transform these constraints into unconstrained forms for the Gaussian and exponential kernels and box constraints for the linear kernel. The advantages of incorporating ordinal information are demonstrated by our simulation results. However, one potential limitation is that we require the latent vector to have length one for the Gaussian and exponential kernels, and length two for the linear kernel, which may affect its generality. For instance, \cite{Qian:2008ts} proposed restricting the correlation to be non-increasing along the ordinal levels, a theoretically rigorous but computationally demanding approach. Further investigation into optimal strategy for incorporating ordinal information is left for future research.

\appendix
\section{Proofs of Theorem~\ref{thm:general_linear} and Propositions~\ref{prop:unique_linear} and \ref{prop:unique_isotropic}}\label{sec:proof}

\begin{proof}[Proof of Theorem~\ref{thm:general_linear}]
Under the linear kernel, the latent parameterization $\{\Wf_{\Vf}\}$ can be obtained through the Cholesky decomposition of the kernel matrix induced by the latent parameterization $\{\Zf_{\Vf}\}$ under the Mercer kernel. 

Specifically, let $\Kf^{(j)}_{\{\Zf_\Vf\}}$ denote the kernel matrix of the $j$th qualitative variable, where the $(v,v^\prime)$th element is given by $ K_j^\Zf(\zf^{(j)}_{v}, \zf^{(j)}_{v^\prime})$ for $v,v^\prime\in\{1,\cdots,a_j\}$. The Cholesky decomposition of $\Kf^{(j)}_{\{\Zf_\Vf\}}$ is expressed as $
\Kf^{(j)}_{\{\Zf_\Vf\}} = \left(\Rf^{(j)}\right)^\top \left(\Rf^{(j)}\right)$, where $\Rf^{(j)} = (\rf^{(j)}_1, \rf^{(j)}_2, \cdots, \rf^{(j)}_{a_j})$ is an upper triangular matrix of dimension $a_j \times a_j$. 
Defining $\wf^{(j)}_{i} = \rf^{(j)}_i$ for $i \in \{1, 2, \cdots, a_j\}$, we then obtain a latent parameterization such that $K_j^\Zf(\zf_{v}^{(j)},\zf_{v^\prime}^{(j)})=K_j^\Wf(\wf_{v}^{(j)},\wf_{v^\prime}^{(j)})$ for all $v,v^\prime\in\{1,\cdots,a_j\}$.

Since the structure between different qualitative variables is either multiplicative as in \eqref{equ:multi_equ} or additive as in \eqref{equ:add_equ}, the equality $K_\Zf\left(\Zf_\Vf, \Zf_{\Vf^\prime}\right) = K_\Wf\left(\Wf_\Vf, \Wf_{\Vf^\prime}\right)$ holds for all $\Vf, \Vf^\prime \in \times_{j=1}^J \{1, 2, \cdots, a_j\}$.
\end{proof}

\begin{proof}[Proof of Proposition~\ref{prop:unique_linear}]  
\textit{Proof of (a).}  
\textit{Existence.} We perform the QR decomposition  
$$
\big(\zf_{1}^{(1)}, \zf_{2}^{(1)},\cdots,\zf_{l_1}^{(1)}\big) = \Qf \Rf,
$$  
where $\Qf \in \Rb^{l_1 \times l_1}$ is an orthogonal matrix, and $\Rf \in \Rb^{l_1 \times l_1}$ is an upper-triangular matrix with positive diagonal entries. Define  
$$
\big(\wf_{1}^{(1)},\wf_{2}^{(1)},\cdots,\wf_{a_1}^{(1)}\big) := \left(\Rf, \Qf^\top \zf_{l_1+1}^{(1)}, \cdots, \Qf^\top \zf_{a_1}^{(1)}\right).
$$  
$\{\Zf_\Vf\}$ and $\{\Wf_\Vf\}$ are equivalent under the linear kernel because they satisfy the desired condition in Lemma~S1. 

\textit{Uniqueness.} We have established the existence of such a latent parameterization. According to Lemma~S1, any equivalent parameterization satisfies
$$
\big(\zf_{1}^{(1)}, \zf_{2}^{(1)}, \cdots, \zf_{a_1}^{(1)}\big) = \Qf \big(\wf_{1}^{(1)}, \cdots, \wf_{a_1}^{(1)}\big)
$$  
for some orthogonal matrix $\Qf\in\Rb^{l_1\times l_1}$. 
The requirements of $\Wf_{\Vf}$ in (a) ensure the uniqueness of the QR decomposition, and consequently $\Qf$, $\wf_{2}^{(1)}, \cdots, \wf_{l_1}^{(1)}$. The uniqueness of the remaining vectors is derived by solving the remaining $a_1 - l_1$ columns of \eqref{equ:unique_proof_iso}.

\textit{Proof of (b).} Cartesian coordinates can be uniquely expressed in hyperspherical coordinates as 
$$
w_{v,l}^{(1)} = r_{v} \cos\left(\theta^{(1)}_{v,l}\right) \prod_{\iota=1}^{l-1} \sin\left(\theta^{(1)}_{v,\iota}\right) \ \text{for}\ 1\leq l<l_1 -1 
\ \text{and}\  
w_{v,l}^{(1)} = r_{v} \prod_{\iota=1}^{l} \sin\left(\theta^{(1)}_{v,\iota}\right)  \ \text{for}\ l=l_1 -1,  
$$  
where $v\in\{1,\cdots,a_1\}$, $0\leq \theta^{(1)}_{v,l}\leq \pi$ for $1<l<l_1-1$ and $0\leq \theta^{(1)}_{v,l}\leq 2\pi$ for $l=l_1 -1$. In the above formula, $ r_{v} = \|\wf_{v}\|_2 = \|\zf_{v}\|_2 = 1 $ for $ 1 \leq v \leq a_1 $. 

To ensure $ w_{v,l}^{(1)} = 0 $ for all $ l > v $, we require $ \theta^{(1)}_{v,l} = 0 $ for $ l = v $. Additionally, we impose $ \theta^{(1)}_{v,l} = 0 $ for $ l > v $ to preserve the equality. Finally, to ensure $ w_{v,v}^{(1)} > 0 $ for $ 1 \leq v \leq l_1 $, we require $ 0 \leq \theta^{(1)}_{v,v-1} \leq \pi $ (rather than less than $ 2\pi $) for $ 1 \leq v \leq l_1 $. These constraints are exactly the ones as described in (b).
\end{proof}

\begin{proof}[Proof of Proposition~\ref{prop:unique_isotropic}]
\textit{Existence.} First, perform the QR decomposition:
$$
\big(\zf_{2}^{(1)} - \zf_{1}^{(1)}, \cdots, \zf_{l_1+1}^{(1)} - \zf_{1}^{(1)}\big) = \Qf \Rf,
$$
where $\Qf \in \Rb^{l_1 \times l_1}$ is an orthogonal matrix, and $\Rf \in \Rb^{l_1 \times l_1}$ is an upper triangular matrix with positive diagonal elements. 
We then define
$$
\big(\wf_{1}^{(1)}, \wf_{2}^{(1)}, \cdots, \wf_{a_1}^{(1)}\big) :=
\left(\mathbf{0}, \Rf, \Qf^\top (\zf_{l_1+2}^{(1)} - \zf_1^{(1)}), \cdots, \Qf^\top (\zf_{a_1}^{(1)} - \zf_1^{(1)}) \right).
$$ 
By setting $\Qf$ and $\wf = -\Qf^\top \zf_1^{(1)}$, $\{\Wf_\Vf\}$ satisfies the desired condition stated in Lemma~S2. Therefore, $\{\Zf_\Vf\}$ and $\{\Wf_\Vf\}$ are equivalent under the isotropic kernel.

\textit{Uniqueness.} We have already proved that such a latent parameterization exists. According to Lemma~S2, there exist an orthogonal matrix $\Qf$ and a vector $\wf$ such that
\begin{equation}\label{equ:unique_proof_iso}
\big(\zf_{1}^{(1)}, \zf_{2}^{(1)}, \cdots, \zf_{a_1}^{(1)}\big) = \Qf \big(\wf_{1}^{(1)} - \wf, \wf_{2}^{(1)} - \wf, \cdots, \wf_{a_1}^{(1)} - \wf\big).
\end{equation}
Since $\wf_{1}^{(1)} = \mathbf{0}$, comparing the first columns of the two matrices in \eqref{equ:unique_proof_iso} gives $\wf = -\Qf^\top \zf_{1}^{(1)}$. Then, we have
$$
\big(\zf_{2}^{(1)}, \cdots, \zf_{l_1+1}^{(1)}\big) = \Qf \big(\wf_{2}^{(1)} + \Qf^\top \zf_{1}^{(1)}, \cdots, \wf_{l_1+1}^{(1)} + \Qf^\top \zf_{1}^{(1)}\big).
$$
Rearranging the above formula yields
$$
\big(\zf_{2}^{(1)} - \zf_{1}^{(1)}, \cdots, \zf_{l_1+1}^{(1)} - \zf_{1}^{(1)}\big) = \Qf \big(\wf_{2}^{(1)}, \cdots, \wf_{l_1+1}^{(1)}\big),
$$
which is precisely the QR decomposition. The uniqueness of $\Qf$, $\wf_{2}^{(1)}, \cdots, \wf_{l_1+1}^{(1)}$, and $\wf$ follows directly from the uniqueness of the QR decomposition, given the requirements for the subdiagonal elements. The uniqueness of the remaining vectors can be derived by solving the remaining $ a_1 - l_1 $ columns of \eqref{equ:unique_proof_iso}. 
\end{proof}

\section*{Disclosure statement}\label{disclosure-statement}
The authors have the following conflicts of interest to declare.

\section*{Acknowledgments}
The authors thank the three reviewers, the Associate Editor, and the Editor for their valuable comments. Deng's work was completed while she was a postdoctoral researcher at the Chinese University of Hong Kong, Shenzhen. 

\bigskip
\begin{center}
{\large\bf SUPPLEMENTARY MATERIAL}
\end{center}
\begin{description}
    \item[Supplementary File] This file includes additional details and results for the numerical comparisons in Section~\ref{sec:num_comp}, additional lemmas, and their proofs.
    \item[Code] This file contains an R package, $\mathsf{MixGP}$, which implements our methods and includes the code to reproduce all simulations, figures, and tables.
\end{description}

\bibliographystyle{asa}
\setlength{\bibsep}{1.5pt}
\bibliography{mybib}
\end{document}


\title{Supplementary Material for ``A General Framework For Modeling Gaussian Process with Qualitative and Quantitative Factors"}
  \author[]{Linsui Deng}
  \author[]{C. F. Jeff Wu}
  \affil[]{School of Data Science, The Chinese University of Hong Kong, Shenzhen, China}
\maketitle
\small 

Our supplement is organized as follows.
Section~\ref{sec:add_nume} provides additional details on the numerical studies, including the simulation settings for benchmark examples and supplementary results for the OTL circuit example, as well as a diagnostic analysis of erroneous inputs in the 3D finite element model. 
Section~\ref{sec:add_simu} presents additional simulation comparisons, including comparisons of different experimental designs and an extensive study of the borehole example under varying discretization levels, with detailed results on prediction accuracy and computational efficiency. 
Section~\ref{sec:para_equ} contains the technical lemmas and their proofs that support the theoretical results in the main paper. 
The \textsf{R} package and reproducible code for this work are also available at
\url{https://github.com/denglinsui/MixGP} and
\url{https://github.com/denglinsui/MixGP-manuscript-sourcecode}, respectively.

\section{Additional details and results of numerical comparisons}\label{sec:add_nume}

\subsection{Functions used in simulation examples}\label{sec:4ex_describe}
In the numerical experiments, we adopt the same setting as presented in \cite{Zhang:2020gp}. We detail the configurations here for completeness and clarity. 
Four benchmark examples are considered in this study: \textit{beam bending}, \textit{borehole}, \textit{OTL circuit}, and \textit{piston}. These examples are commonly used for evaluating the performance of surrogate modeling techniques because they are derived from real-world engineering problems and carry significant practical relevance. The functional forms of these four examples are originally defined based on continuous input variables. To adapt them to our study, we discretize certain continuous variables and treat them as ordinal qualitative variables.

The \textit{beam bending} example models the deflection of a beam with an elastic modulus of $E = 600 \, \text{GPa}$. The beam is fixed at one end, with a vertical force of $P = 600 \, \text{N}$ applied at the free end. The deflection at the free end depends on the beam length $L\in[10,20]$, the height of the beam's cross-section $h\in[1,2]$, and the cross-sectional shape $v$, and is given by:
$$
y(L, h, v) = \frac{L^3}{3 \times 10^9 h^4 I(v)},
$$
where $I(v)$ represents the normalized moment of inertia depending on the cross-section type. In this example, six cross-section types are considered, including circular, square, I-shape, hollow square, hollow circular and  H-shape, and the corresponding $I(v)$ are approximately 0.0491, 0.0833, 0.0449, 0.0633, 0.0373, and 0.0167 respectively. The geometric pictures can be found in \cite{Zhang:2020gp}. 

The \textit{borehole} example considers the flow of water through a borehole that is drilled from the ground surface through two aquifers \citep{harper1983sensitivity,morris1993bayesian}. The flow rate (in $m^3/\text{year}$) through the borehole is given by:
$$
y(r_w,r,T_u,H_u,T_l,H_l,L,K_w) = \frac{2 \pi T_u \left(H_u - H_l\right)}{ \log \left({r}/{r_w}\right) \left\{ 1 + \frac{2 L T_u}{\log \left({r}/{r_w}\right) r_w^2 K_w} + \frac{T_u}{T_l} \right\}},
$$
where $r_w \in [0.05, 0.15]$ (in $m$) is the radius of the borehole, $r \in [100, 50000]$ (in $m$) is the radius of influence, $T_u \in [63070, 115600]$ (in $m^2/\text{year}$) is the transmissivity of the upper aquifer, $H_u \in [990, 1110]$ (in $m$) is the potentiometric head of the upper aquifer, $T_l \in [63.1, 116]$ (in $m^2/\text{year}$) is the transmissivity of the lower aquifer, $H_l \in [700, 820]$ (in $m$) is the potentiometric head of the lower aquifer, $L \in [1120, 1680]$ (in $m$) is the length of the borehole, and $K_w \in [9855, 12045]$ (in $m/\text{year}$) is the hydraulic conductivity of the borehole lining.

Following \cite{Zhang:2020gp}, we adapt this function to evaluate surrogate models with mixed inputs by treating $r_w$ (radius of the borehole) and $H_l$ (potentiometric head of the lower aquifer) as qualitative factors. Specifically, $r_w$ is discretized into three levels $\{0.05, 0.10, 0.15\}$, and $H_l$ is discretized into four levels $\{700 , 740 , 780 , 820 \}$. 

The \textit{output transformerless (OTL) circuit} example considers the midpoint voltage of a push-pull output transformerless amplifier circuit. This circuit is commonly used in electrical engineering to drive loads without the need for a transformer \citep{ben2007modeling}. The midpoint voltage (in volts) is modeled as:
$$
\begin{aligned}
y(R_{b1}, R_{b2}, R_f, R_{c1}, R_{c2}, B) &=  \frac{\beta\left(V_{b1} + 0.74\right)\left(R_{c2} + 9\right)}{\beta \left(R_{c2} + 9\right) + R_f} + \frac{11.35 R_f}{\beta \left(R_{c2} + 9\right) + R_f} \\
&\qquad+  \frac{0.74 \beta R_f(R_{c2} + 9)}{R_{c1}\left\{\beta \left(R_{c2} + 9\right) + R_f\right\}},
\end{aligned}
$$
where $V_{b1} ={12 R_{b2}}/(R_{b1} + R_{b2})$, $R_{b1} \in [50, 150]$ (in $k\Omega$) is the resistance $b1$, $R_{b2} \in [25, 70]$ (in $k\Omega$) is the resistance$b2$, $R_f \in [0.5, 3]$ (in $k\Omega$) is the resistance $f$, $R_{c1} \in [1.2, 2.5]$ (in $k\Omega$) is the resistance $c1$, $R_{c2} \in [0.25, 1.2]$ (in $k\Omega$) is the resistance $c2$, and $\beta \in [50, 300]$ (Amperes) is the current gain of the transistor.

Following \cite{Zhang:2020gp}, we adapt this example to evaluate surrogate models with mixed inputs by treating $R_f$ (resistance $f$) and $\beta$ (current gain) as qualitative factors. Specifically, $R_f$ is discretized into four levels $\{0.5, 1.2, 2.1, 2.9\}$, and $\beta$ is discretized into six levels $\{50, 100, 150, 200, 250, 300\}$. 

The \textit{piston} example models the cycle time of a piston moving within a cylinder. This example simulates the linear motion of a piston being transformed into rotational motion via a connected rod and disk. The cycle time (in seconds) is a measure of the piston’s performance and is affected by multiple physical factors. The response function for the cycle time is given by
\begin{equation}
y(M, S, V_0, k, P_0, T_a, T_0) = 2 \pi \sqrt{\frac{M}{k + S^2 \frac{P_0 V_0 T_a}{V^2 T_0}}},
\end{equation}
where
\begin{equation}
V = \frac{S}{2k} \sqrt{A^2 + 4k \frac{P_0 V_0}{T_0} T_a-A} 
\quad\text{and}\quad 
A = P_0 S + 19.62 M - \frac{k V_0}{S},
\end{equation}
$M \in [30, 60]$ (in kg) represents the mass of the piston, $S \in [0.005, 0.020]$ (in $m^2$) represents the surface area of the piston, $V_0 \in [0.002, 0.010]$ (in $m^3$) is the initial gas volume, $k \in [1000, 5000]$ (in $N/m$) is the spring coefficient, $P_0 \in [90000, 110000]$ (in $N/m^2$) is the atmospheric pressure, $T_a \in [290, 296]$ (in $K$) is the ambient temperature, and $T_0 \in [340, 360]$ (in $K$) is the filling gas temperature.

Following \cite{Zhang:2020gp}, this example is adapted to evaluate surrogate models with mixed inputs by treating $P_0$ (atmospheric pressure) and $k$ (spring coefficient) as qualitative factors. Specifically, $P_0$ is discretized into three levels $\{90000, 100000, 110000\}$, and $k$ is discretized into five levels $\{1000, 2000, 3000, 4000, 5000\}$. 

\subsection{Additional results of the OTL examples}

Figure~\ref{fig:LV_ORD_ADD} visualizes the latent vectors $z_1$ for resistance $R_f$ and current gain $\beta$, estimated by $\Gau^{\add}_{1}$, $\Gau^{\add}_{\ord}$, $\Exp^{\add}_{1}$, and $\Exp^{\add}_{\ord}$, respectively, across 30 replications in the OTL example. Similar to Figure~\ref{fig:LV_ORD} in the paper, incorporating ordinal information preserves the natural ordinal structure and benefits from regularization effects.

\begin{figure}
    \centering
    \subfigure[$R_f$, additive Gaussian kernel.]{
        \includegraphics[width=0.45\linewidth]{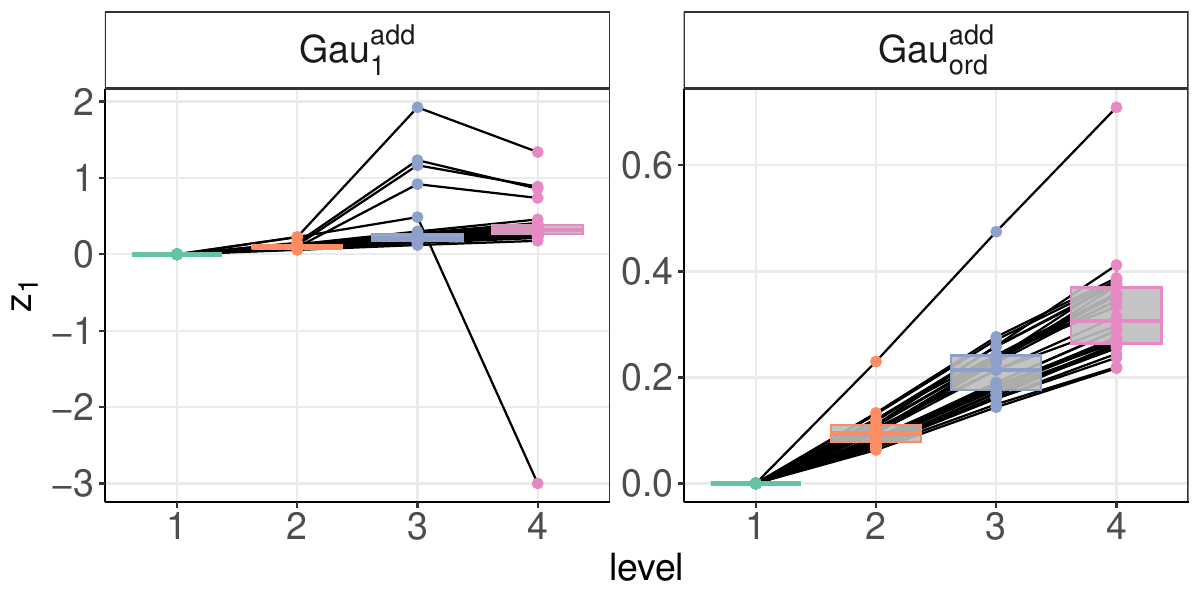}
        \label{fig:lv_ord_rf_gau}
    }
    \subfigure[$R_f$, additive exponential kernel.]{
        \includegraphics[width=0.45\linewidth]{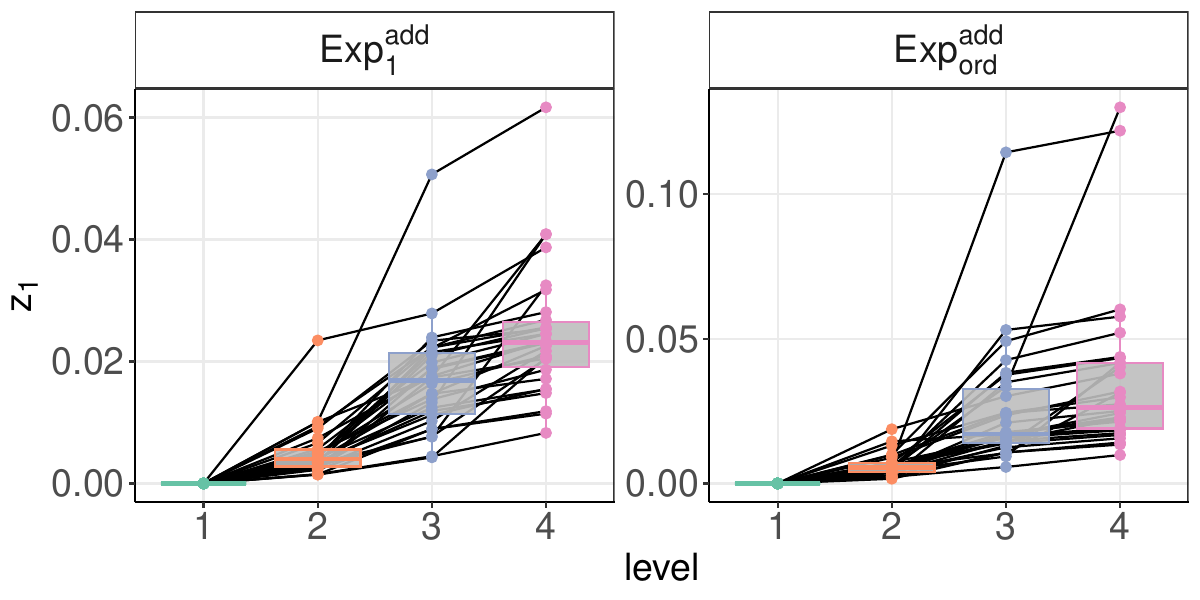}
        \label{fig:lv_ord_rf_exp}
    }
    \vskip\baselineskip 
    \subfigure[$\beta$, additive Gaussian kernel.]{
        \includegraphics[width=0.45\linewidth]{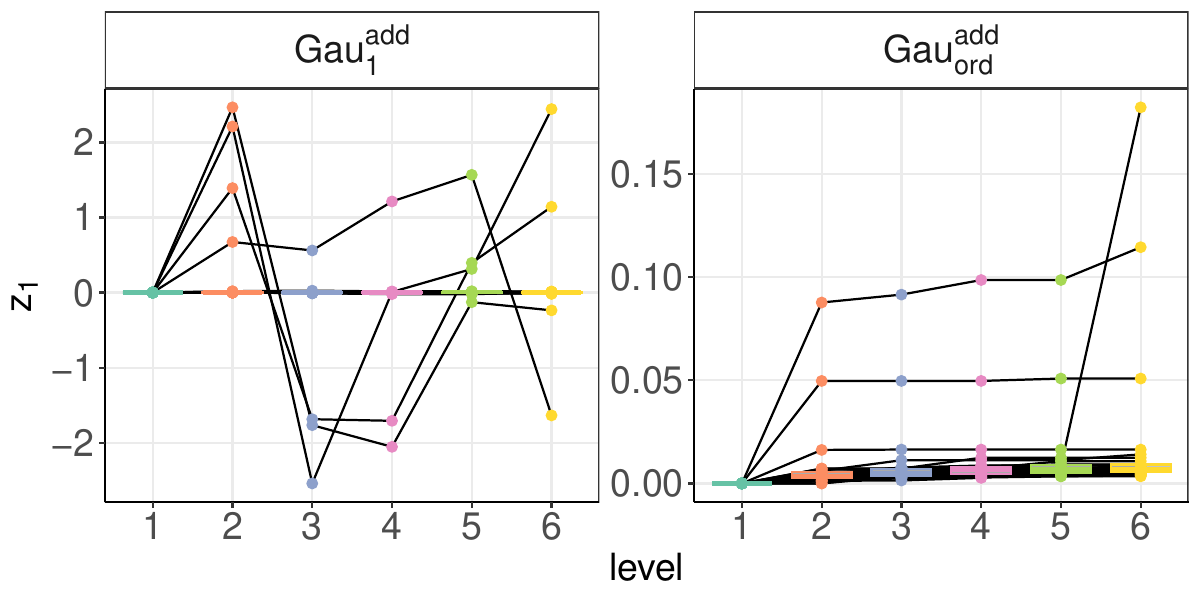}
        \label{fig:lv_ord_beta_gau}
    }
    \subfigure[$\beta$, additive exponential kernel.]{
        \includegraphics[width=0.45\linewidth]{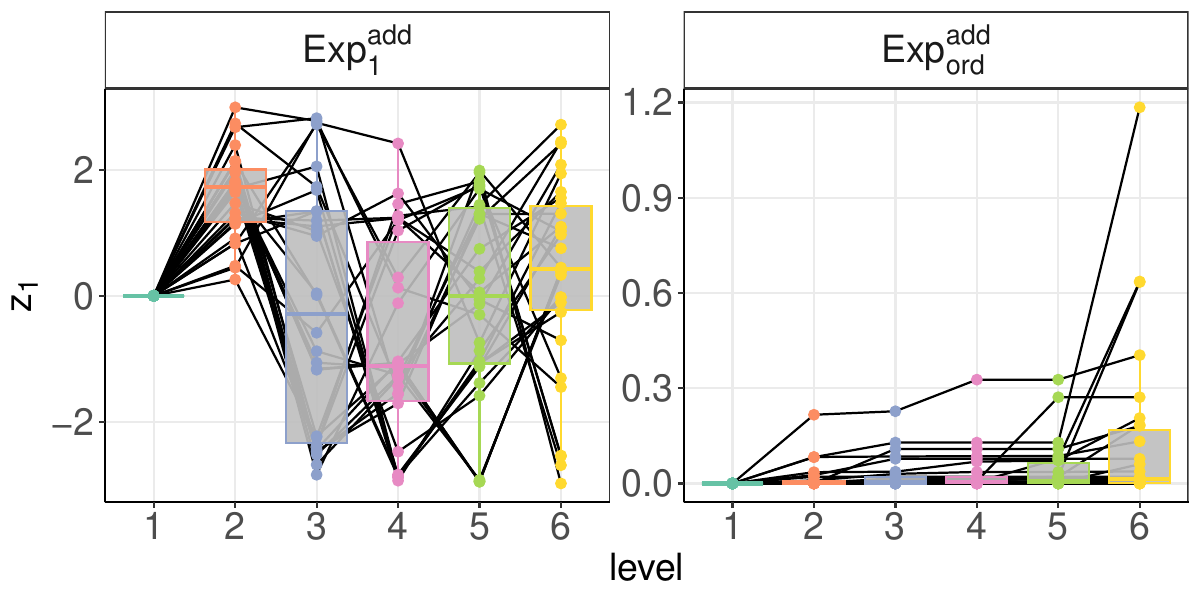}
        \label{fig:lv_ord_beta_exp}
    }
    \caption{The boxplots and scatter points depict the latent vectors $z_1$ for resistance $R_f$ and current gain $\beta$, estimated by $\Gau^{\add}_{1}$, $\Gau^{\add}_{\ord}$, $\Exp^{\add}_{1}$, and $\Exp^{\add}_{\ord}$, respectively, across 30 replications in the OTL example. Points from the same replication are connected by lines. }
    \label{fig:LV_ORD_ADD}
\end{figure}

\subsection{Diagnosis of erroneous inputs in the 3D finite element model}\label{sec:3D_err}
In Section~\ref{sec:3D} of our main paper, we point out two data points, $(u_1, v_1, v_2, v_3)\in$$\{(14,5,200,8000)$, $(14,10,200,1578)\}$, potentially having reversed labels and exclude them from the training process. Based on the fact that there are only 13 valid combinations for $(v_1, v_2, v_3)$ $\in$ $\{(1,100,4800)$, $(5,200,1578)\}$, we suspect the correct inputs should be $(u_1, v_1, v_2, v_3)$ $\in$ $\{(14,1,100,4800)$, $(14,5,200,1578)\}$.

To verify our suspicion, we use the 18 base models trained in Section~\ref{sec:3D} to make predictions for both the original inputs---$(14,5,200,8000)$ (Location 1) and $(14,10,200,1578)$ (Location 2)---and the revised inputs---$(14,1,100,4800)$ (Location 1) and $(14,5,200,1578)$ (Location 2). The predicted values are summarized in Table~\ref{tab:predicted_3}.

The observed value at Location 1 is $-0.535$, and at Location 2, it is $-0.768$. We find that the predictions by using the revised inputs are much closer to the observed values compared to those by using the original inputs. This confirms that the two points are indeed outliers and should be excluded from the training process.

\begin{table}[h!]
\centering
\caption{Predicted values from 18 base models for both the original and revised inputs at two locations identified as potential outliers. Results are summarized as the median, mean, and standard deviation (SD) across replications.}
\label{tab:predicted_3}
\resizebox{\textwidth}{!}{%
\begin{tabular}{cccccccccccccccccccc}
\toprule
&&  \multicolumn{18}{c}{{Method}}  \\
\cmidrule{3-20}
  $n$&Criterion& \multicolumn{3}{c}{$\Gau^{\multi}$}   &\multicolumn{3}{c}{$\Exp^{\multi}$}   & \multicolumn{3}{c}{$\Linear^{\multi}$}  &  \multicolumn{3}{c}{$\Gau^{\add}$}    & \multicolumn{3}{c}{$\Exp^{\add}$}     &\multicolumn{3}{c}{$\Linear^{\add}$}   \\
\cmidrule(r){3-5} \cmidrule(lr){6-8} \cmidrule(lr){9-11} \cmidrule(lr){12-14} \cmidrule(lr){15-17} \cmidrule(lr){18-20} 
  && 1-d & 2-d &$\ord$ & 1-d & 2-d &$\ord$& 1-d & 2-d &$\ord$& 1-d & 2-d &$\ord$& 1-d & 2-d &$\ord$& 1-d & 2-d &$\ord$\\
\midrule
 \multicolumn{20}{c}{Location 1 (Original Input)} \\
\midrule
 & Mean & -0.756 & -0.756 & -0.747 & -0.750 & -0.748 & -0.751 & -0.742 & -0.744 & -0.748 & -0.727 & -0.721 & -0.721 & -0.738 & -0.732 & -0.730 & -0.732 & -0.732 & -0.731\\
\cmidrule{2-20}
27 & Median & -0.754 & -0.754 & -0.749 & -0.754 & -0.755 & -0.754 & -0.746 & -0.749 & -0.750 & -0.729 & -0.725 & -0.725 & -0.742 & -0.737 & -0.738 & -0.731 & -0.735 & -0.734\\
\cmidrule{2-20}
& SD & 0.017 & 0.015 & 0.010 & 0.026 & 0.034 & 0.010 & 0.038 & 0.028 & 0.011 & 0.047 & 0.048 & 0.042 & 0.051 & 0.050 & 0.045 & 0.063 & 0.061 & 0.052\\
\midrule
& Mean &  -0.752 & -0.753 & -0.748 & -0.749 & -0.749 & -0.750 & -0.750 & -0.754 & -0.751 & -0.722 & -0.719 & -0.717 & -0.731 & -0.728 & -0.726 & -0.723 & -0.722 & -0.721\\
\cmidrule{2-20}
45 & Median & -0.751 & -0.753 & -0.749 & -0.752 & -0.752 & -0.755 & -0.752 & -0.754 & -0.752 & -0.718 & -0.717 & -0.716 & -0.730 & -0.726 & -0.726 & -0.718 & -0.718 & -0.717\\
\cmidrule{2-20}
 & SD & 0.010 & 0.009 & 0.007 & 0.010 & 0.009 & 0.009 & 0.033 & 0.010 & 0.008 & 0.032 & 0.032 & 0.024 & 0.031 & 0.031 & 0.028 & 0.038 & 0.037 & 0.033\\
\midrule
 & Mean & -0.752 & -0.754 & -0.750 & -0.750 & -0.750 & -0.752 & -0.752 & -0.757 & -0.754 & -0.711 & -0.708 & -0.709 & -0.718 & -0.716 & -0.714 & -0.709 & -0.709 & -0.709\\
\cmidrule{2-20}
63 & Median & -0.750 & -0.755 & -0.750 & -0.753 & -0.754 & -0.756 & -0.753 & -0.756 & -0.755 & -0.709 & -0.708 & -0.708 & -0.715 & -0.713 & -0.713 & -0.709 & -0.709 & -0.709\\
\cmidrule{2-20}
 & SD & 0.008 & 0.008 & 0.005 & 0.009 & 0.009 & 0.009 & 0.030 & 0.025 & 0.006 & 0.020 & 0.017 & 0.016 & 0.021 & 0.020 & 0.018 & 0.016 & 0.016 & 0.015\\
\midrule
 \multicolumn{20}{c}{Location 1 (Revised Input)} \\
\midrule
 & Mean & -0.531 & -0.531 & -0.533 & -0.536 & -0.536 & -0.533 & -0.531 & -0.533 & -0.533 & -0.578 & -0.575 & -0.575 & -0.595 & -0.587 & -0.584 & -0.597 & -0.595 & -0.594\\
\cmidrule{2-20}
27 & Median & -0.533 & -0.533 & -0.534 & -0.535 & -0.536 & -0.534 & -0.534 & -0.534 & -0.534 & -0.559 & -0.558 & -0.559 & -0.575 & -0.567 & -0.567 & -0.564 & -0.566 & -0.560\\
\cmidrule{2-20}
 & SD & 0.007 & 0.008 & 0.007 & 0.007 & 0.007 & 0.006 & 0.011 & 0.007 & 0.008 & 0.064 & 0.060 & 0.056 & 0.069 & 0.061 & 0.055 & 0.090 & 0.084 & 0.088\\
\midrule
 & Mean & -0.534 & -0.535 & -0.535 & -0.537 & -0.537 & -0.536 & -0.535 & -0.536 & -0.535 & -0.577 & -0.570 & -0.572 & -0.586 & -0.583 & -0.581 & -0.578 & -0.580 & -0.576\\
\cmidrule{2-20}
45 & Median & -0.535 & -0.535 & -0.535 & -0.536 & -0.536 & -0.535 & -0.535 & -0.535 & -0.535 & -0.559 & -0.554 & -0.555 & -0.570 & -0.563 & -0.563 & -0.559 & -0.559 & -0.560\\
\cmidrule{2-20}
 & SD & 0.005 & 0.005 & 0.007 & 0.004 & 0.004 & 0.004 & 0.007 & 0.006 & 0.007 & 0.064 & 0.066 & 0.058 & 0.061 & 0.063 & 0.061 & 0.066 & 0.064 & 0.059\\
\midrule
 & Mean & -0.536 & -0.536 & -0.536 & -0.536 & -0.536 & -0.537 & -0.537 & -0.537 & -0.537 & -0.552 & -0.550 & -0.551 & -0.558 & -0.558 & -0.556 & -0.554 & -0.554 & -0.554\\
\cmidrule{2-20}
63 & Median & -0.536 & -0.536 & -0.536 & -0.536 & -0.536 & -0.536 & -0.536 & -0.536 & -0.537 & -0.550 & -0.549 & -0.550 & -0.557 & -0.554 & -0.555 & -0.552 & -0.552 & -0.551\\
\cmidrule{2-20}
 & SD & 0.002 & 0.002 & 0.003 & 0.002 & 0.002 & 0.002 & 0.004 & 0.003 & 0.005 & 0.019 & 0.017 & 0.016 & 0.020 & 0.019 & 0.017 & 0.017 & 0.017 & 0.017\\
\midrule
 \multicolumn{20}{c}{Location 2 (Original Input)} \\
\midrule
 & Mean & -0.841 & -0.841 & -0.829 & -0.841 & -0.841 & -0.843 & -0.811 & -0.822 & -0.820 & -0.882 & -0.881 & -0.881 & -0.880 & -0.881 & -0.879 & -0.877 & -0.877 & -0.878\\
\cmidrule{2-20}
27 & Median & -0.841 & -0.839 & -0.833 & -0.838 & -0.838 & -0.838 & -0.828 & -0.830 & -0.827 & -0.876 & -0.878 & -0.877 & -0.876 & -0.879 & -0.874 & -0.876 & -0.876 & -0.871\\
\cmidrule{2-20}
 & SD & 0.028 & 0.026 & 0.028 & 0.019 & 0.020 & 0.017 & 0.058 & 0.035 & 0.029 & 0.051 & 0.052 & 0.048 & 0.050 & 0.051 & 0.051 & 0.062 & 0.063 & 0.058\\
\midrule
 & Mean & -0.818 & -0.823 & -0.802 & -0.832 & -0.830 & -0.834 & -0.812 & -0.821 & -0.803 & -0.883 & -0.884 & -0.883 & -0.882 & -0.884 & -0.883 & -0.883 & -0.883 & -0.883\\
\cmidrule{2-20}
45 & Median & -0.827 & -0.831 & -0.788 & -0.833 & -0.832 & -0.836 & -0.817 & -0.823 & -0.816 & -0.876 & -0.877 & -0.876 & -0.876 & -0.875 & -0.876 & -0.874 & -0.874 & -0.874\\
\cmidrule{2-20}
 & SD & 0.029 & 0.027 & 0.027 & 0.015 & 0.015 & 0.013 & 0.030 & 0.023 & 0.035 & 0.035 & 0.033 & 0.032 & 0.037 & 0.035 & 0.036 & 0.040 & 0.040 & 0.039\\
\midrule
& Mean & -0.808 & -0.821 & -0.799 & -0.831 & -0.831 & -0.834 & -0.820 & -0.826 & -0.822 & -0.879 & -0.878 & -0.878 & -0.877 & -0.878 & -0.878 & -0.878 & -0.878 & -0.878\\
\cmidrule{2-20}
63 & Median & -0.811 & -0.835 & -0.790 & -0.835 & -0.834 & -0.837 & -0.819 & -0.829 & -0.827 & -0.876 & -0.876 & -0.875 & -0.874 & -0.876 & -0.876 & -0.875 & -0.875 & -0.875\\
\cmidrule{2-20}
 & SD & 0.031 & 0.029 & 0.028 & 0.013 & 0.014 & 0.010 & 0.033 & 0.027 & 0.022 & 0.027 & 0.026 & 0.026 & 0.028 & 0.027 & 0.027 & 0.026 & 0.026 & 0.026\\
\midrule
 \multicolumn{20}{c}{Location 2 (Revised Input)} \\
\midrule
 & Mean & -0.773 & -0.773 & -0.773 & -0.775 & -0.775 & -0.775 & -0.775 & -0.774 & -0.772 & -0.826 & -0.821 & -0.822 & -0.827 & -0.825 & -0.823 & -0.817 & -0.817 & -0.819\\
\cmidrule{2-20}
27 & Median &-0.769 & -0.769 & -0.769 & -0.770 & -0.769 & -0.770 & -0.770 & -0.768 & -0.768 & -0.821 & -0.814 & -0.817 & -0.825 & -0.825 & -0.823 & -0.814 & -0.814 & -0.813\\
\cmidrule{2-20}
 & SD & 0.015 & 0.016 & 0.015 & 0.017 & 0.020 & 0.015 & 0.017 & 0.017 & 0.015 & 0.049 & 0.048 & 0.044 & 0.049 & 0.048 & 0.050 & 0.059 & 0.059 & 0.056\\
\midrule
& Mean & -0.768 & -0.768 & -0.768 & -0.769 & -0.768 & -0.769 & -0.766 & -0.765 & -0.765 & -0.822 & -0.821 & -0.821 & -0.825 & -0.825 & -0.824 & -0.820 & -0.818 & -0.819\\
\cmidrule{2-20}
45 & Median & -0.766 & -0.766 & -0.766 & -0.766 & -0.767 & -0.766 & -0.765 & -0.765 & -0.765 & -0.817 & -0.815 & -0.815 & -0.817 & -0.819 & -0.818 & -0.813 & -0.812 & -0.813\\
\cmidrule{2-20}
 & SD & 0.010 & 0.010 & 0.009 & 0.013 & 0.014 & 0.011 & 0.012 & 0.010 & 0.006 & 0.034 & 0.033 & 0.031 & 0.036 & 0.033 & 0.035 & 0.039 & 0.040 & 0.039\\
\midrule
 & Mean & -0.767 & -0.766 & -0.766 & -0.766 & -0.766 & -0.768 & -0.766 & -0.766 & -0.766 & -0.818 & -0.816 & -0.816 & -0.821 & -0.819 & -0.818 & -0.815 & -0.815 & -0.814\\
\cmidrule{2-20}
63 & Median & -0.768 & -0.769 & -0.768 & -0.769 & -0.769 & -0.768 & -0.767 & -0.769 & -0.768 & -0.816 & -0.813 & -0.813 & -0.819 & -0.818 & -0.818 & -0.811 & -0.811 & -0.811\\
\cmidrule{2-20}
 & SD & 0.012 & 0.012 & 0.012 & 0.011 & 0.012 & 0.007 & 0.008 & 0.012 & 0.009 & 0.026 & 0.025 & 0.025 & 0.029 & 0.027 & 0.025 & 0.026 & 0.026 & 0.025\\
\bottomrule
\end{tabular}
}
\end{table}

\section{Additional simulation comparisons}\label{sec:add_simu}
\subsection{Comparisons of different experimental designs}\label{sec:comp_design}
In this section, we examine the effect of experimental design on model performance. Specifically, we consider random sampling, the MaxPro design \citep{joseph2015maximum,joseph2020designing}, and the maximin Latin hypercube design (LHD) \citep{santner2003design}, which was previously examined in Section~\ref{sec:simu}, each using 80 training points. As shown in Figure~\ref{fig:differdesign}, the RRMSE values are highly consistent across all three designs for every method, indicating that the choice of design has a negligible impact on performance.

\begin{figure}
    \centering
    \includegraphics[width=\linewidth]{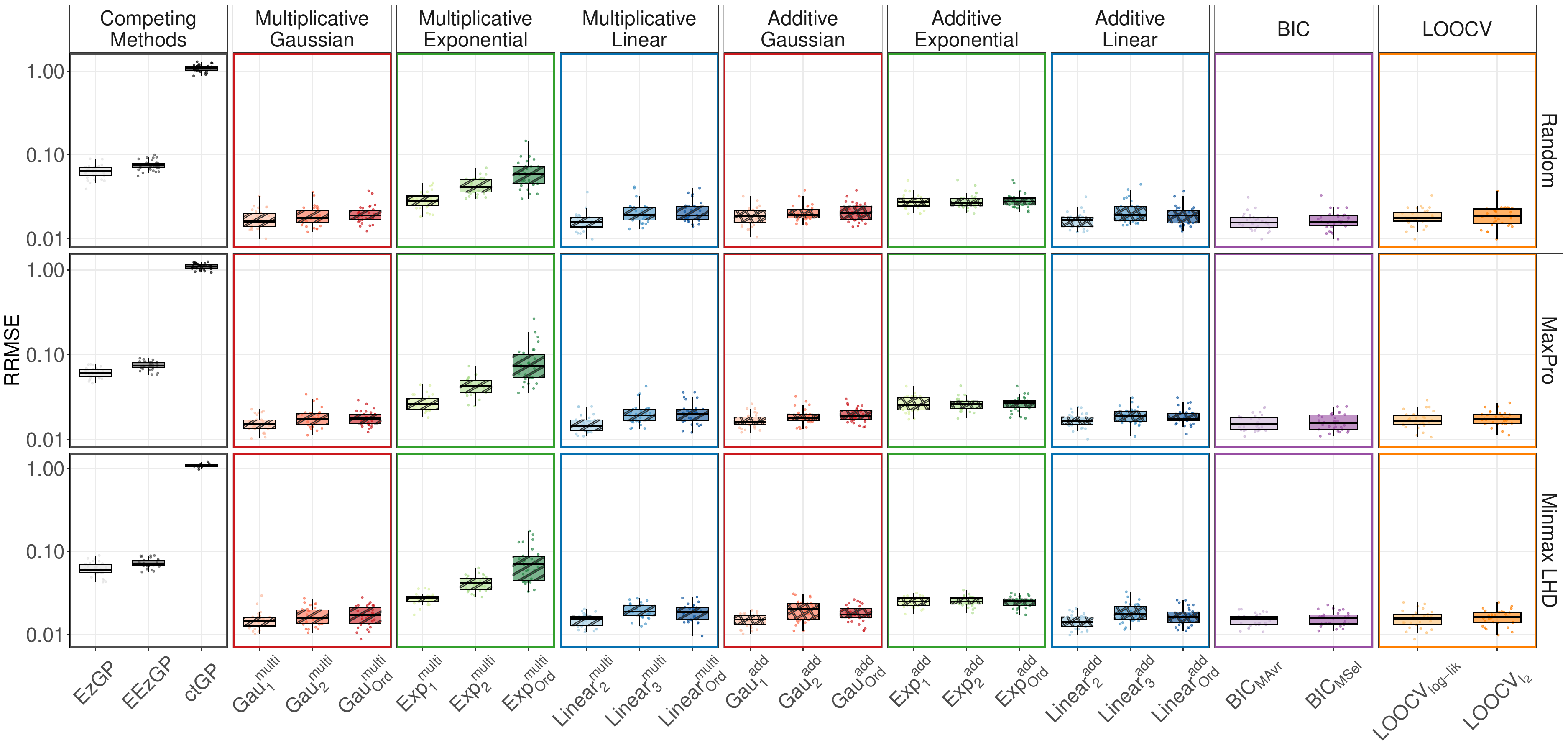}
    \caption{Comparison of RRMSE across different methods and kernel configurations for the OTL circuit example. Each boxplot summarizes the results from 30 independent runs with different training sets generated via minimax LHS, random, and MaxPro designs, respectively. All methods were evaluated on the same set of 10,000 uniformly distributed test points.}
    \label{fig:differdesign}
\end{figure}

\subsection{The borehole example with varying discretization degree}\label{sec:borehore_vary_level}

In this section, we evaluate the predictive accuracy and computational efficiency of various methods for the borehole example under different discretization levels.

In the preliminary example presented in Section~\ref{sec:simu}, $r_w$ was discretized into $q_1 = 3$ levels and $H_l$ into $q_2 = 4$ levels. Here, we extend our analysis by considering multiple combinations of discretization degrees: $q_1 \in \{2,4,6\}$ for $r_w$ and $q_2 \in \{2,4,6,8,10\}$ for $H_l$. We generate $5 q_1 q_2$ design points using a minimax LHD, and map $r_w \in [0.05,0.15]$ and $H_l \in [700,820]$ to equally spaced values. To save computation time, we run our methods with only three random restarts and select the solution yielding the smallest log‑likelihood. 

Figures~\ref{fig:S5_2}--\ref{fig:S5_6} report the RRMSEs of all methods. 
Consistent with the results in Section~\ref{sec:simu}, methods with additive kernels (including the competing methods $\EzGP$ and $\EEzGP$) exhibit poor predictive performance. Methods with multiplicative Gaussian and exponential kernels generally perform well. In contrast, methods with multiplicative linear kernels deliver satisfactory results only under very coarse discretization ($q_1=2$, $q_2=2$) if ordinal information is not leveraged, possibly because fewer restarts increase the risk of getting stuck in local optima. By comparison, $\Linear_\ord^\multi$, together with $\Gau_\ord^\multi$, generally exhibits the best performance, which indicates that incorporating ordinal information can substantially enhance predictive accuracy and also improve optimization stability.

Table~\ref{tab:time} summarizes the computational times of three competing methods and our 18 base methods. The computational time of model selection and model averaging approaches are not reported, as they are constructed by post-processing and aggregating the results from our base models, rather than being standalone methods with with minimal additional computational costs. All computational times were measured on a personal computer equipped with a 13th Gen Intel(R) Core(TM) i7-13700KF CPU (3.40\,GHz), 64\,GB RAM, running Windows~10 and \textsf{R}~4.4.1. $\EEzGP$ is the fastest overall. 
$\ctGP$ is efficient when the number of discretization levels is small, but its computational time increases rapidly as the levels grow. 
$\EzGP$ has complexity comparable to that of our base methods. 
For the base methods, computational time generally increases with latent dimension. 
With Gaussian or exponential kernels, methods incorporating ordinal information 
($\Gau_{\ord}^\multi$, $\Gau_{\ord}^\add$, $\Exp_{\ord}^\multi$, and $\Exp_{\ord}^\add$) exhibit computational time comparable to the corresponding methods with the same latent dimension $l_j=1$ ($\Gau_{1}^\multi$, $\Gau_{1}^\add$, $\Exp_{1}^\multi$, and $\Exp_{1}^\add$), and can even be faster when the discretization degree is high. 
However, with multiplicative linear kernels, incorporating ordinal information greatly increases costs.  
Despite the higher computational time, $\Linear_{\ord}^\multi$ achieves higher prediction accuracy than $\Linear_{2}^\multi$ or $\Linear_{1}^\multi$. Thus, there is a trade-off between accuracy and efficiency.

Figure~\ref{fig:RRMSE_time_borehole} displays the RRMSE and computational time trends of 
$\EzGP$, $\EEzGP$, $\ctGP$, $\Gau_{\ord}^\multi$, $\Linear_{\ord}^\multi$, and $\Linear_{2}^\multi$. 
Among our methods, $\Linear_{\ord}^\multi$ and $\Linear_{2}^\multi$ represent the slowest and fastest in terms of computational time, respectively; from the perspective of prediction accuracy, $\Gau_{\ord}^\multi$ consistently ranks among the top-performing base methods; and we include $\BIC_\MAvr$ as an exemplar of the four model averaging and model selection methods because they yield similar predictive accuracy in the reported RRMSE results. $\EzGP$ is the fastest method, but the predictive accuracy of $\EzGP$ and $\EEzGP$ is relatively low. 
Although finer discretization requires estimating correlations among more levels, the RRMSE of $\ctGP$ decreases as larger sample sizes offset the added complexity. This gain in accuracy comes at the cost of reduced computational efficiency.  
For our methods, RRMSE decreases at low to moderate discretization degrees and increases when the degree is high. 
Finally, $\Gau_{\ord}^\multi$ delivers high predictive accuracy while maintaining reasonable computational time.

  \begin{figure}
    \centering
    \includegraphics[width=\linewidth]{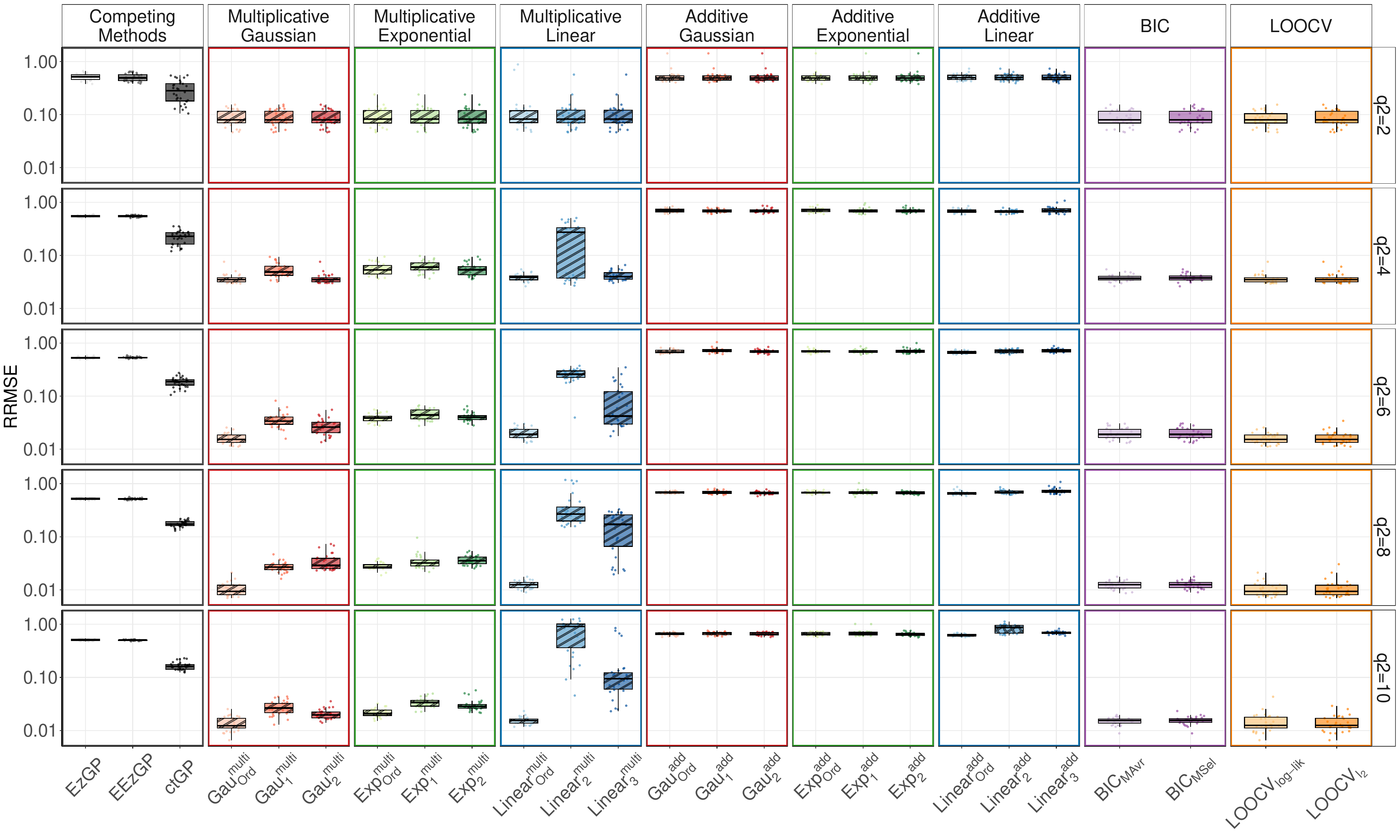}
    \caption{Comparison of RRMSE across different methods and kernel configurations for the borehole example with discretization levels $q_1=2$ and $q_2 \in \{2,4,6,8,10\}$. Each boxplot summarizes the results from $5 q_1q_2$ independent runs with different training sets generated via maximin Latin hypercube designs. All methods were evaluated on the same set of 10,000 uniformly distributed test points.}
    \label{fig:S5_2}
\end{figure}

  \begin{figure}
    \centering
    \includegraphics[width=\linewidth]{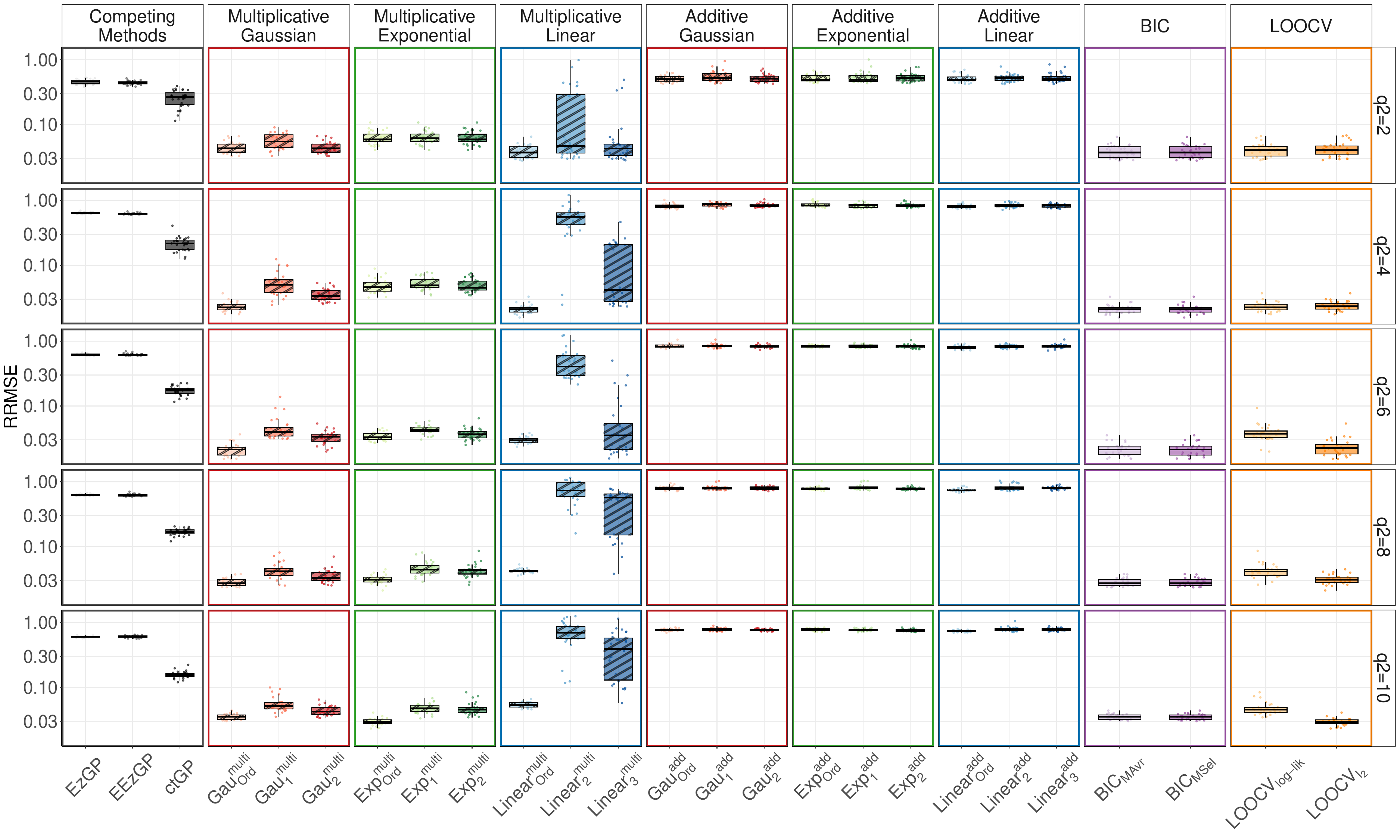}
    \caption{Comparison of RRMSE across different methods and kernel configurations for the borehole example with discretization levels $q_1 =4$ and $q_2 \in \{2,4,6,8,10\}$, using the same experimental setup as in Figure~\ref{fig:S5_2}.}
    \label{fig:S5_4}
\end{figure}

  \begin{figure}
    \centering
    \includegraphics[width=\linewidth]{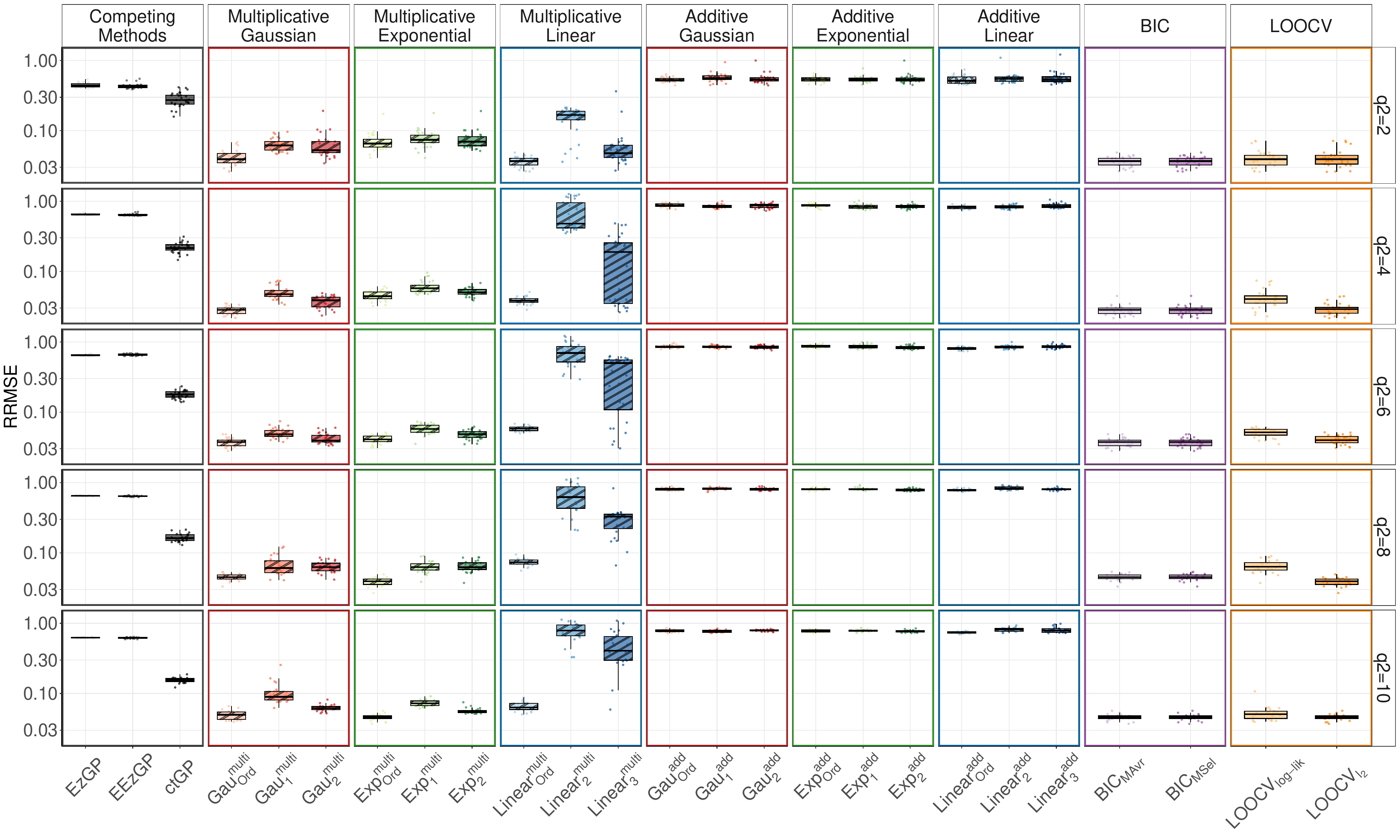}
    \caption{Comparison of RRMSE across different methods and kernel configurations for the borehole example with discretization levels $q_1=6$ and $q_2 \in \{2,4,6,8,10\}$, using the same experimental setup as in Figure~\ref{fig:S5_2}.}
    \label{fig:S5_6}
\end{figure}

\begin{table}[h!]
\centering
\caption{Computational time (in seconds) for three competing methods and our 18 base methods applied to the borehole problem under different discretization levels, using the same experimental setup as in Figure~\ref{fig:S5_2}.}
\label{tab:time}
\resizebox{\textwidth}{!}{%
\begin{tabular}{ccccccccccccccccccccccc}
\hline
&&  \multicolumn{21}{c}{{Method}}  \\
\cmidrule{3-23}
$q_2$ & Time & \multirow{2}{*}{$\EzGP$} & \multirow{2}{*}{$\EEzGP$} & \multirow{2}{*}{$\ctGP$} & \multicolumn{3}{c}{$\Gau^{\multi}$}   &\multicolumn{3}{c}{$\Exp^{\multi}$}   & \multicolumn{3}{c}{$\Linear^{\multi}$}  &  \multicolumn{3}{c}{$\Gau^{\add}$}    & \multicolumn{3}{c}{$\Exp^{\add}$}     &\multicolumn{3}{c}{$\Linear^{\add}$}  \\
\cmidrule(r){6-8} \cmidrule(lr){9-11} \cmidrule(lr){12-14} \cmidrule(lr){15-17} \cmidrule(lr){18-20} \cmidrule(lr){21-23}
  &&&&& 1-d & 2-d &$\ord$ & 1-d & 2-d &$\ord$& 1-d & 2-d &$\ord$& 1-d & 2-d &$\ord$& 1-d & 2-d &$\ord$& 1-d & 2-d &$\ord$\\
\midrule
 \multicolumn{23}{c}{$q_1=2$} \\
 \midrule
\multirow{2}{*}{2} & Mean & 5 & 2 & 9 & 29 & 29 & 29 & 29 & 29 & 29 & 21 & 21 & 33 & 29 & 29 & 29 & 29 & 29 & 29 & 22 & 22 & 35\\
\cmidrule{2-23}
 & SD & 1 & $<$1 & 3 & 1 & 1 & 1 & 1 & 1 & 1 & 1 & 1 & 3 & 2 & 2 & 2 & 2 & 2 & 2 & 2 & 2 & 6\\
\midrule
\multirow{2}{*}{4} & Mean & 21 & 5 & 11 & 38 & 61 & 50 & 40 & 68 & 41 & 44 & 53 & 109 & 54 & 71 & 52 & 55 & 73 & 51 & 45 & 53 & 82\\
\cmidrule{2-23}
 & SD & 3 & 1 & 3 & 4 & 7 & 5 & 4 & 11 & 4 & 6 & 8 & 28 & 9 & 8 & 7 & 9 & 10 & 7 & 6 & 6 & 25\\
\midrule
\multirow{2}{*}{6} & Mean & 47 & 12 & 14 & 75 & 110 & 90 & 61 & 131 & 65 & 54 & 83 & 287 & 67 & 97 & 74 & 72 & 100 & 77 & 60 & 79 & 92\\
\cmidrule{2-23}
 & SD & 3 & 3 & 2 & 12 & 5 & 7 & 8 & 11 & 8 & 9 & 8 & 67 & 14 & 16 & 11 & 10 & 24 & 10 & 11 & 9 & 13\\
\midrule
\multirow{2}{*}{8} & Mean & 98 & 21 & 22 & 110 & 162 & 138 & 114 & 201 & 115 & 85 & 125 & 650 & 105 & 156 & 104 & 114 & 172 & 122 & 78 & 125 & 131\\
\cmidrule{2-23}
 & SD & 1 & 4 & 2 & 13 & 9 & 6 & 11 & 12 & 11 & 10 & 8 & 152 & 14 & 11 & 16 & 17 & 26 & 10 & 15 & 15 & 24\\
\midrule
\multirow{2}{*}{10} & Mean & 178 & 38 & 37 & 162 & 245 & 195 & 163 & 302 & 185 & 104 & 190 & 712 & 144 & 228 & 148 & 152 & 269 & 153 & 107 & 185 & 177\\
\cmidrule{2-23}
 & SD & 13 & 9 & 2 & 21 & 10 & 11 & 14 & 18 & 10 & 14 & 12 & 183 & 29 & 22 & 27 & 40 & 31 & 30 & 22 & 13 & 35\\
\midrule
 \multicolumn{23}{c}{$q_1=4$} \\
 \midrule
\multirow{2}{*}{2} & Mean & 20 & 6 & 10 & 37 & 51 & 47 & 40 & 63 & 40 & 37 & 42 & 70 & 45 & 62 & 46 & 46 & 63 & 46 & 38 & 43 & 67\\
\cmidrule{2-23}
 & SD & 1 & 1 & 2 & 4 & 6 & 5 & 3 & 5 & 3 & 5 & 6 & 10 & 6 & 6 & 6 & 6 & 7 & 7 & 6 & 7 & 15\\
\midrule
\multirow{2}{*}{4} & Mean & 91 & 18 & 24 & 75 & 143 & 116 & 93 & 168 & 79 & 69 & 110 & 398 & 87 & 130 & 89 & 100 & 137 & 100 & 73 & 102 & 120\\
\cmidrule{2-23}
 & SD & 9 & 5 & 3 & 13 & 17 & 13 & 14 & 18 & 12 & 16 & 16 & 116 & 18 & 25 & 11 & 11 & 22 & 13 & 13 & 13 & 29\\
\midrule
\multirow{2}{*}{6} & Mean & 231 & 48 & 63 & 159 & 264 & 203 & 169 & 312 & 209 & 116 & 193 & 968 & 152 & 244 & 141 & 171 & 275 & 176 & 109 & 186 & 186\\
\cmidrule{2-23}
 & SD & 6 & 14 & 4 & 25 & 22 & 23 & 15 & 26 & 28 & 19 & 20 & 290 & 31 & 19 & 29 & 31 & 38 & 23 & 26 & 32 & 49\\
\midrule
\multirow{2}{*}{8} & Mean & 490 & 89 & 182 & 302 & 467 & 324 & 257 & 481 & 285 & 165 & 314 & 790 & 202 & 408 & 207 & 247 & 461 & 274 & 161 & 311 & 352\\
\cmidrule{2-23}
 & SD & 10 & 26 & 13 & 51 & 70 & 58 & 51 & 100 & 37 & 42 & 42 & 143 & 56 & 28 & 50 & 62 & 47 & 56 & 39 & 28 & 85\\
\midrule
\multirow{2}{*}{10} & Mean & 872 & 166 & 376 & 426 & 688 & 365 & 412 & 720 & 358 & 274 & 478 & 875 & 401 & 657 & 414 & 465 & 756 & 451 & 345 & 521 & 392\\
\cmidrule{2-23}
 & SD & 22 & 36 & 15 & 102 & 107 & 75 & 65 & 99 & 62 & 55 & 95 & 191 & 69 & 69 & 35 & 85 & 71 & 30 & 60 & 72 & 102\\
\midrule
 \multicolumn{23}{c}{$q_1=6$} \\
 \midrule
\multirow{2}{*}{2} & Mean & 50 & 13 & 14 & 63 & 109 & 89 & 75 & 125 & 76 & 57 & 78 & 138 & 74 & 98 & 76 & 75 & 112 & 84 & 63 & 80 & 104\\
\cmidrule{2-23}
 & SD & 2 & 3 & 2 & 10 & 4 & 4 & 9 & 8 & 9 & 9 & 10 & 49 & 14 & 17 & 11 & 9 & 22 & 9 & 10 & 15 & 23\\
\midrule
\multirow{2}{*}{4} & Mean & 228 & 39 & 65 & 164 & 272 & 213 & 185 & 308 & 183 & 109 & 190 & 892 & 151 & 247 & 139 & 168 & 264 & 167 & 110 & 175 & 200\\
\cmidrule{2-23}
 & SD & 5 & 13 & 4 & 29 & 22 & 21 & 11 & 25 & 18 & 22 & 19 & 226 & 35 & 14 & 33 & 37 & 33 & 30 & 28 & 28 & 66\\
\midrule
\multirow{2}{*}{6} & Mean & 630 & 92 & 261 & 324 & 582 & 359 & 309 & 594 & 333 & 191 & 388 & 774 & 248 & 482 & 258 & 311 & 585 & 278 & 216 & 383 & 407\\
\cmidrule{2-23}
 & SD & 12 & 31 & 19 & 72 & 77 & 67 & 38 & 88 & 55 & 54 & 70 & 192 & 63 & 19 & 71 & 71 & 69 & 70 & 52 & 33 & 77\\
\midrule
\multirow{2}{*}{8} & Mean & 1273 & 208 & 713 & 550 & 975 & 522 & 494 & 1115 & 472 & 402 & 660 & 1110 & 533 & 879 & 541 & 605 & 1035 & 613 & 464 & 698 & 589\\
\cmidrule{2-23}
 & SD & 80 & 51 & 32 & 170 & 204 & 114 & 103 & 257 & 80 & 89 & 149 & 244 & 106 & 85 & 40 & 93 & 144 & 43 & 65 & 124 & 163\\
\midrule
\multirow{2}{*}{10} & Mean & 2051 & 327 & 1542 & 998 & 1485 & 845 & 993 & 1797 & 848 & 674 & 1145 & 2902 & 753 & 1293 & 756 & 836 & 1534 & 873 & 567 & 999 & 867\\
\cmidrule{2-23}
 & SD & 101 & 86 & 107 & 213 & 295 & 195 & 185 & 301 & 149 & 91 & 214 & 493 & 152 & 96 & 148 & 99 & 104 & 66 & 119 & 160 & 172\\
\hline
\end{tabular}}
\end{table}

\begin{figure}[htbp]
    \centering
    \subfigure[Computational time for selected methods ($\EzGP$, $\EEzGP$, $\ctGP$, $\Gau_\ord^\multi$, $\Linear_\ord^\multi$ and $\Linear_2^\multi$).]{
        \includegraphics[width=\linewidth]{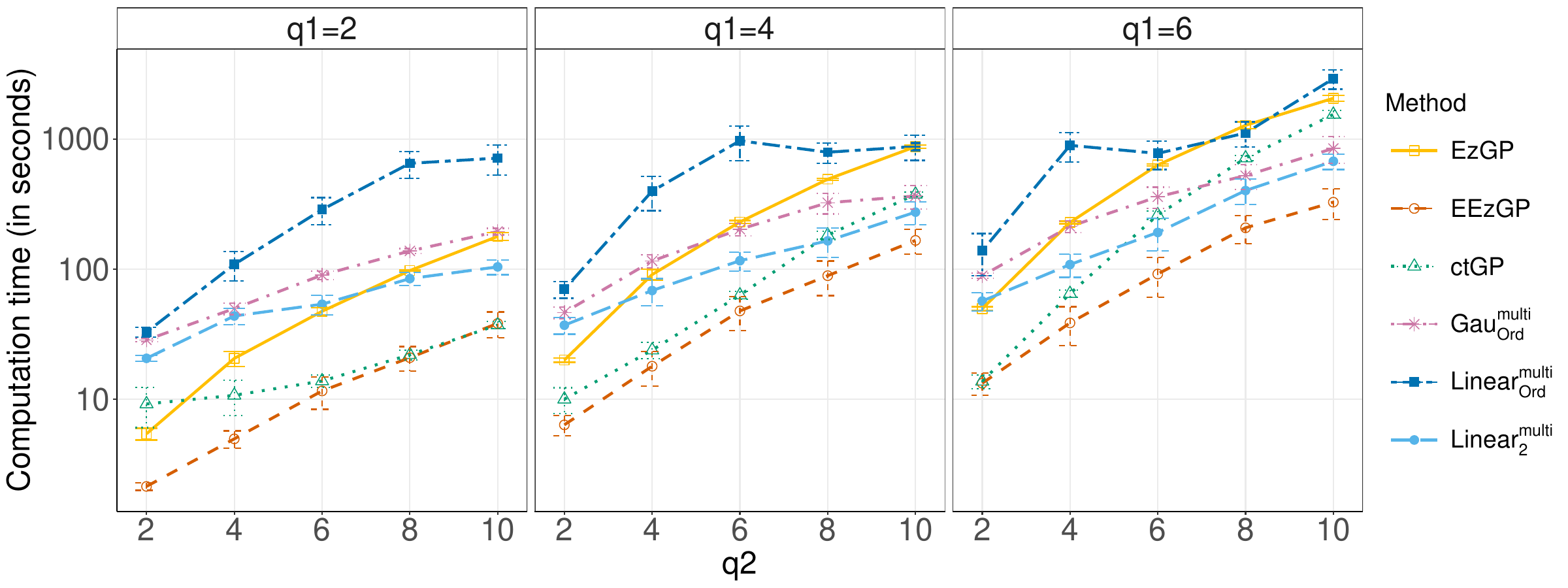}
        \label{fig:time}
    }
    \hfill
    \subfigure[RRMSE for selected methods, and additionally including the BIC-based model averaging method ($\BIC_\MAvr$).]{
        \includegraphics[width=\linewidth]{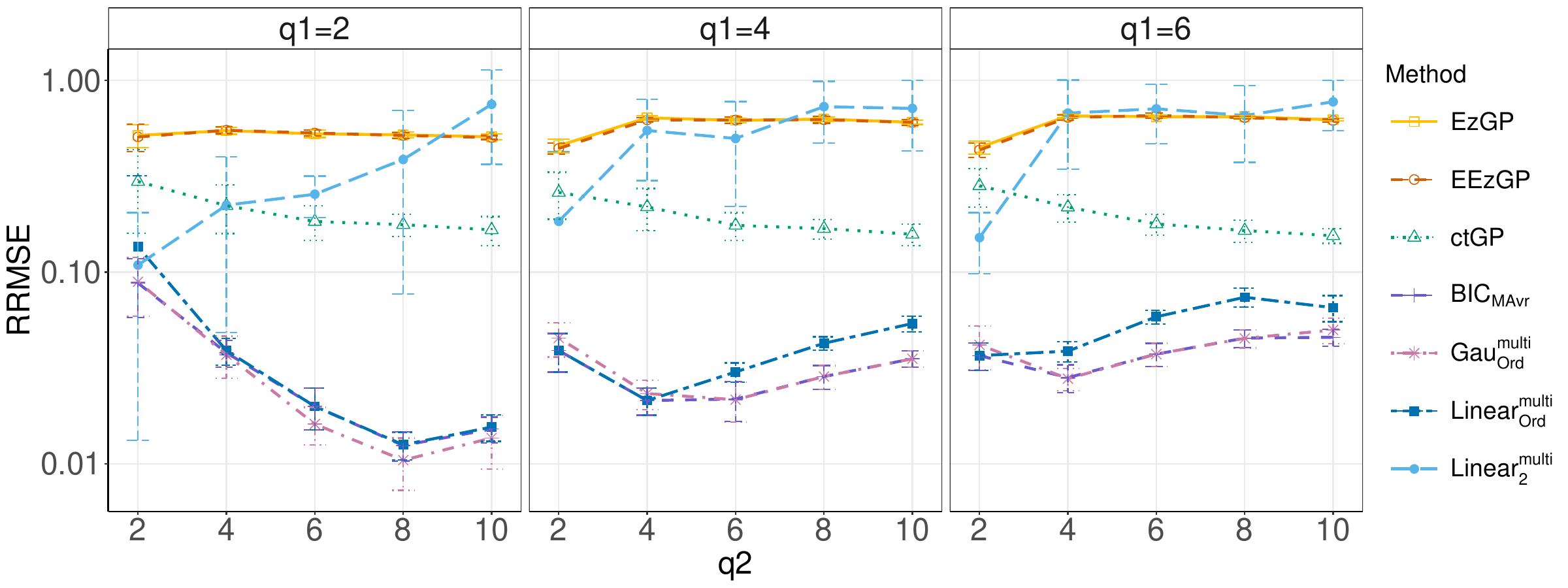}
        \label{fig:RRMSE_borehole}
    }
    \caption{
Comparison of computational time (in seconds) and prediction accuracy for the borehole example across varying discretization levels, using the same experimental setup as in Figure~\ref{fig:S5_2}.
}
    \label{fig:RRMSE_time_borehole}
\end{figure}

\section{Technical lemmas and their proofs}\label{sec:para_equ}

In this section, we demonstrate the parametric equivalence under the conditions $I = 0$, $J = 1$, and when the two kernels $K_{\Zf}(\cdot, \cdot)$ and $K_{\Wf}(\cdot, \cdot)$ are identical. Under these assumptions, the equivalence of parameterizations depends solely on the latent parameterizations $\{\Zf_\Vf\}$ and $\{\Wf_\Vf\}$.

For the linear kernel, Lemma~\ref{prop:equi_linear} states that the equivalence of the latent parameterizations is determined up to an orthogonal transformation.  
\begin{lemma}[Equivalence under Linear Kernel] \label{prop:equi_linear}
Two latent parameterizations $\left\{\Wf_\Vf\right\}$ and $\left\{\Zf_\Vf\right\}$ are equivalent under the linear kernel if and only if  $\big(\zf_{1}^{(1)},\zf_{2}^{(1)},\cdots,\zf_{a_1}^{(1)}\big)$ $=$ $\Qf\big(\wf_{1}^{(1)},\wf_{2}^{(1)},\cdots,\wf_{a_1}^{(1)}\big)$ for some orthogonal matrix $\Qf\in\Rb^{l_1\times l_1}$ such that $\Qf^\top \Qf=\If_{l_1\times l_1}$. 
\end{lemma}

\begin{proof}[Proof of Lemma~\ref{prop:equi_linear}]
    Denote $\Zf^{(1)} = \big(\zf_{1}^{(1)}, \zf_{2}^{(1)}, \cdots, \zf_{a_1}^{(1)}\big)$ and $\Wf^{(1)} = \big(\wf_{1}^{(1)}, \wf_{2}^{(1)}, \cdots, \wf_{a_1}^{(1)}\big)$. 
If there exists an orthogonal matrix $\Qf \in \Rb^{l_1 \times l_1}$ such that $\Zf^{(1)} = \Qf \Wf^{(1)}$, then we have
$$
\left(\Wf^{(1)}\right)^\top \Wf^{(1)}
=
\left(\Wf^{(1)}\right)^\top \Qf^\top \Qf \Wf^{(1)}
=
\left(\Zf^{(1)}\right)^\top \Zf^{(1)},
$$
which implies that $\wf_{v}^\top \wf_{v^\prime} = \zf_{v}^\top \zf_{v^\prime}$ for any $v, v^\prime \in \{1, \cdots, a_1\}$.

For the reverse direction, we perform the QR decomposition as follows
$$
\Zf^{(1)} = \Qf_{\Zf} \Rf_{\Zf}
\quad\text{and}\quad 
\Wf^{(1)} = \Qf_{\Wf} \Rf_{\Wf}
$$
where $\Qf_{\Zf}, \Qf_{\Wf} \in \Rb^{l_1 \times l_1}$ are orthogonal matrices, $\Rf_{\Zf}$ and $\Rf_{\Wf}$ are upper triangular matrices with positive diagonal elements. 
Since $\left\{\Wf_\Vf\right\}$ and $\left\{\Zf_\Vf\right\}$ are equivalent under the linear kernel, we have:
$$
\Rf_{\Wf}^\top \Rf_{\Wf} =
\left(\Wf^{(1)}\right)^\top \Wf^{(1)} =
\left(\Zf^{(1)}\right)^\top \Zf^{(1)} =
\Rf_{\Zf}^\top \Rf_{\Zf}.
$$
According to the Cholesky decomposition and its uniqueness, we deduce that $\Rf_{\Wf} = \Rf_{\Zf}$. The orthogonal matrix can be defined as $\Qf = \Qf_\Zf \Qf_\Wf^\top$.
\end{proof}

For an isotropic kernel, the kernel value is uniquely determined by the distance between two latent vectors. Consequently, any isometric transformation of $\Rb^{l_1}$ results in an equivalent latent parameterization, which is formalized in Lemma~\ref{prop:equi_station}.

\begin{lemma}[Equivalence under Isotropic Kernel] \label{prop:equi_station}
Two latent parameterizations $\left\{\Wf_\Vf\right\}$ and $\left\{\Zf_\Vf\right\}$ are equivalent under an isotropic kernel if and only if 
$$
\big(\zf_{1}^{(1)}, \zf_{2}^{(1)}, \cdots, \zf_{a_1}^{(1)}\big) = \Qf \big(\wf_{1}^{(1)} - \wf, \wf_{2}^{(1)} - \wf, \cdots, \wf_{a_1}^{(1)} - \wf\big),
$$
for some vector $\wf \in \Rb^{l_1}$ and some orthogonal matrix $\Qf \in \Rb^{l_1 \times l_1}$ such that $\Qf^\top \Qf = \If_{l_1 \times l_1}$.
\end{lemma}

\begin{proof}[Proof of Lemma~\ref{prop:equi_station}]
Let $\Zf^{(1)} = \big(\zf_{1}^{(1)}, \zf_{2}^{(1)}, \cdots, \zf_{a_1}^{(1)}\big)$ and $\Wf^{(1)} = \big(\wf_{1}^{(1)}, \wf_{2}^{(1)}, \cdots, \wf_{a_1}^{(1)}\big)$. Suppose there exists an orthogonal matrix $\Qf \in \Rb^{l_1 \times l_1}$ such that  
\begin{equation}\label{equ:iso_target}
\big(\zf_{1}^{(1)}, \zf_{2}^{(1)}, \cdots, \zf_{a_1}^{(1)}\big) = \Qf \big(\wf_{1}^{(1)} - \wf, \wf_{2}^{(1)} - \wf, \cdots, \wf_{a_1}^{(1)} - \wf\big),
\end{equation}
where $\wf \in \Rb^{l_1}$. In this case, parameter equivalence is guaranteed by
$$
K_\Zf\left(\Zf_\Vf, \Zf_{\Vf^\prime}\right) =  
\Kcal_1\left(\left\|\zf^{(1)}_{v_1} - \zf^{(1)}_{v_1^\prime}\right\|_2\right) =  
\Kcal_1\left(\left\|\Qf\left(\wf^{(1)}_{v_1} - \wf^{(1)}_{v_1^\prime}\right)\right\|_2\right) =  
K_{\Wf}\left(\Wf_\Vf, \Wf_{\Vf^\prime}\right),
$$
where $\Kcal_1(\cdot)$ is the decreasing kernel generating function. 

On the other hand, assume the above formula holds for any $v_1, v_1^\prime \in \{1, \cdots, a_1\}$. Then, we have
$$
\left\|\zf^{(1)}_{v_1} - \zf^{(1)}_{v_1^\prime}\right\|^2_2 = \left\|\wf^{(1)}_{v_1} - \wf^{(1)}_{v_1^\prime}\right\|^2_2,
$$
since $\Kcal_1(\cdot)$ is a decreasing function. Note that Euclidean distance can be expressed using the inner product. By Lemma~\ref{prop:equi_linear}, we know there exists an orthogonal matrix $\Qf \in \Rb^{l_1 \times l_1}$ such that
$$
\zf^{(1)}_{v_1} - \zf^{(1)}_{v_1^\prime} = \Qf\left(\wf^{(1)}_{v_1} - \wf^{(1)}_{v_1^\prime}\right),
$$
for any $v_1, v_1^\prime \in \{1, \cdots, a_1\}$. Thus, defining $\wf = \Qf^\top \zf^{(1)}_{v_1} - \wf^{(1)}_{v_1}$ results in a constant vector $\wf$, independent of $v_1$, based on the formula above. Finally, the defined $\Qf$ and $\wf$ ensure that \eqref{equ:iso_target} holds, which justifies the reverse direction as well. This completes the proof.
\end{proof}

\bibliographystyle{asa}
\setlength{\bibsep}{1.5pt}
\bibliography{mybib}